\documentclass{amsart}
\usepackage[vcentermath,enableskew]{youngtab}
\usepackage{ytableau}
\usepackage{amssymb, amsmath, graphicx, amsthm}
\usepackage{thmtools}
\usepackage{hyperref}
\usepackage[all,cmtip]{xy}
\usepackage{multicol}
\usepackage{enumerate}
\usepackage{tikz}
\usepackage{color}

\definecolor{pole}{RGB}{255,192,203}
\definecolor{zero}{RGB}{135,206,250}
\definecolor{doublezero}{RGB}{152,251,152}

\newcommand\abs[1]{\lvert#1\rvert}

\DeclareMathOperator{\es}{\emptyset}
\DeclareMathOperator{\ee}{\epsilon}
\DeclareMathOperator{\R}{\mathbb{R}}

\DeclareMathOperator{\C}{\mathbb{C}}
\DeclareMathOperator{\Z}{\mathbb{Z}}
\DeclareMathOperator{\N}{\mathbb{N}}
\DeclareMathOperator{\Q}{\mathbb{Q}}
\DeclareMathOperator{\E}{\mathbb{E}}
\DeclareMathOperator{\ZZ}{\mathcal{Z}}

\newcommand{\la}{\langle}
\newcommand{\ra}{\rangle}
\newcommand{\TT}{\mathbb{T}}
\newcommand{\tT}{\tilde{T}}
\DeclareMathOperator{\RR}{\mathcal{R}}
\DeclareMathOperator{\Hom}{\text{Hom}}
\DeclareMathOperator{\End}{\text{End}}
\DeclareMathOperator{\Res}{\text{Res}}

\renewcommand{\Re}{\operatorname{Re}}

\newtheorem{lem}{Lemma}
\newtheorem{thm}{Theorem}
\newtheorem{claim}{Claim}

\theoremstyle{definition}

\newtheorem{rmk}{Remark}
\begin{document}
\title{Analyticity of Nekrasov Partition Functions}
\author[G. Felder]{Giovanni Felder}
\author[M. M\"uller-Lennert]{Martin M\"uller-Lennert}
\address{Department of mathematics, ETH Zurich, 8092 Zurich, Switzerland}
\email{felder@math.ethz.ch}
\email{martin.mueller-lennert@math.ethz.ch}
\begin{abstract} We prove that the K-theoretic Nekrasov instanton
  partition functions have a positive radius of convergence in the
  instanton counting parameter and are holomorphic functions of the
  Coulomb parameters in a suitable domain. We discuss the implications
  for the AGT correspondence and the analyticity of the norm of
  Gaiotto states for the deformed Virasoro algebra. The proof is based
  on random matrix techniques and relies on an integral representation
  of the partition function, due to Nekrasov, which we also prove.
\end{abstract}
\unitlength = 1mm
\maketitle
\tableofcontents

\section{Introduction}
The purpose of this paper is to study the analytic properties of the
Nekrasov instanton partition function.  In $\mathcal N=2$
supersymmetric gauge theory in four and five dimensions, Nekrasov's
instanton partition function \cite{Nekrasov2002} plays the role of a
basic building block. In the physical interpretation, the instanton
partition function is the non-perturbative contribution of instantons
to a gauge theory in an $\Omega$-background.  Mathematically it is the
generating function of integrals of torus equivariant cohomology
classes ($K$-theory classes in the five dimensional theory) on the
moduli space of framed torsion free sheaves with fixed rank on the
complex projective plane. The torus is a product of a two-dimensional
torus acting on the projective plane and a torus acting on the
framing.  The parameters $\epsilon_1,\epsilon_2$ of the
$\Omega$-background are the equivariant parameters in the mathematical
description and serve as an infrared regulator. One early success of
the theory \cite{Nekrasov2002,NekrasovOkounkov2006} was a microscopic
justification of the Seiberg-Witten formula \cite{SeibergWitten1994}
for the prepotential of the low energy effective theory, which arises
in the limit $\epsilon_1,\epsilon_2\to 0$.  The Nekrasov instanton
partition function is a power series in a complex variable
$\mathfrak q$ parametrizing the gauge theory.  The coefficient of
$\mathfrak{q}^n$ is the contribution of instanton number $n$ to the
partition function. The simplest case of a Nekrasov partition function
appears in the pure $\mathcal N=2$ supersymmetric Yang--Mills theory
with gauge group $U(r)$ on $\mathbb R^4\times S^1$.  It is given as a
sum over $r$-tuples $\vec Y=(Y_i)_{i=1}^r$ of Young diagrams of total
size $|\vec Y|=n$:
\begin{align*}
  Z(\epsilon_1,\epsilon_2,a,\mathfrak{q},\lambda)
  &=\sum_{\vec Y}\mathfrak{q}^{|\vec Y|}
    \prod_{\alpha,\beta=1}^r\prod_{b\in Y_\alpha}\frac{(\lambda/2)^{2}}
    {\sinh\left(\frac\lambda2 E_{\alpha\beta}(b)\right)
    \sinh\left(\frac\lambda2 (\epsilon_1+\epsilon_2-E_{\alpha\beta}(b))\right)},
  \\
  E_{\alpha\beta}(b)&=a_\alpha-a_\beta-l_{Y_\beta}(b)\epsilon_1+(a_{Y_\alpha}(b)+1)\epsilon_2.
\end{align*}
Here $\lambda$ is the circumference of the circle $S^1$ and
$a=(a_1,\dots,a_r)$ belongs to the Lie algebra of $U(r)$ and
parametrizes boundary conditions of scalar fields in the
$\mathcal N=2$ vector multiplet.  The arm length $a_Y(b)$ of a box $b$
is the number of boxes in the Young diagram $Y$ to the right of $b$;
the leg length $l_Y(b)$ is the (possibly negative) number of boxes in
$Y$ below $b$. The four dimensional theory on $\mathbb R^4$ arises in
the limit $\lambda\to 0$ and amounts to replacing
$\sinh (\lambda x)/\lambda$ by $x$ and $\mathfrak q=\Lambda^{2r}$ is related
to the dynamical mass scale $\Lambda$ of the gauge theory. In the mathematical description,
we view $\mathbb R^4$ as $\mathbb C^2$ with its action of
$(\mathbb C^\times)^2$ and embed it in $\mathbb {CP}^2$ by adding a
line at infinity $\ell_\infty\cong\mathbb{CP}^1$.  Let
$\mathcal M(r,n)$ be the moduli space of torsion free sheaves on
$\mathbb {CP}^2$ of rank $r$ and second Chern class $n$ with a
framing, i.e., a trivialization on $\ell_\infty$. It is a smooth
algebraic variety of dimension $2nr$. The action of the group
$(\mathbb C^\times)^2$ lifts to an action on $\mathcal
M(r,n)$. Moreover $\mathit{GL}_r$, and in particular its Cartan torus
$T=(\mathbb C^\times)^{r}$ acts on $\mathcal M(r,n)$ by changing the
framing. Therefore we have an action of
$\tilde T=(\mathbb C^\times)^2\times T$ on $\mathcal M(r,n)$.  For a
semisimple $\tilde T$-module $V$ with finite dimensional weight spaces
$V_\chi$, we denote by
$\mathrm{ch}(V)=\sum_{\chi} e^\chi\mathrm{dim}\,V_\chi $ its formal
character in the completed group ring of the weight lattice.  It is a
formal series in the weight variables $q_1,q_2$ for the
$(\mathbb C^\times)^2$-action and $u_1,\dots,u_r$ for the
$T$-action. Then
\begin{equation}\label{e-powerseries}
  Z(\epsilon_1,\epsilon_2,a,\mathfrak{q},\lambda)=
  \sum_{n=0}^\infty z^n
  \sum_{i=1}^{2nr}(-1)^i\mathrm{ch}\,H^i(\mathcal M(r,n),\mathcal O),
\quad
z=\mathfrak{q}\lambda^{2r}e^{-\lambda(\epsilon_1+\epsilon_2)r/2},
\end{equation}
is the pushforward of the $\tilde T$-equivariant $K$-theory class of
the structure sheaf by the map to a point, with the identification
\[
  q_1=e^{-\lambda\epsilon_1}.\quad q_2=e^{-\lambda\epsilon_2},\quad
  u_\alpha=e^{-\lambda a_\alpha},\quad \alpha=1,\dots,r. 
\]
The combinatorial formula above is obtained by localization to fixed
points. There are many variants and generalizations of the Nekrasov
partition functions. In this paper we focus on 
the case where one include matter fields in the
gauge theory, which corresponds to replacing the structure sheaf by
more general $K$-theory classes in the mathematical description.

Besides being the instanton contribution to the partition function of
a gauge theory in the $\Omega$-background, $|Z|^2$, in the limit
$\lambda\to 0$, appears in the integrand of Pestun's formula for gauge
theory on the round sphere $S^4$ \cite{Pestun2012} and in its
extension to $S^4$ with an ellipsoid metric parametrized by
$\epsilon_1,\epsilon_2>0$ \cite{HamaHosomichi2012}.  For $\lambda>0$
this formula extends to $S^4\times S^1$ and other compact manifolds,
see \cite{Pasquetti2016} for a review. For these reason it is
important to understand the convergence properties of the formal power
series defining the partition function.  In this paper we prove that
the power series has a positive radius of convergence if $\lambda>0$
and in a suitable range of parameters. For example, in the case
of pure gauge theory, we show that if
$\lambda>0$ and $\epsilon_1,\epsilon_2>0$, the series
\eqref{e-powerseries} converges for all $z$ in the unit disk to an
analytic function of the Coulomb parameters $a_i$ in a neighbourhood
of the imaginary axis, see Theorem \ref{thm:nekrasovconvergence} for
the general result. Note that for this range of parameters a direct
estimation of the sum over partitions is problematic because of small
denominators.  In fact the individual terms are not defined if
$\epsilon_1/\epsilon_2\in\pi\mathbb Q$ and have arbitrarily small
denominators otherwise. The situation is different in the case
$\epsilon_1<0<\epsilon_2$ (which we do not consider) where a
positive
radius of convergence may be obtained by direct estimates on
individual terms for generic $a_i$. Our proof applies to the case
$\epsilon_1,\epsilon_2>0$ and also $\epsilon_1=\bar\epsilon_2$ not
imaginary, and relies on an integral representation of the
coefficients of $Z$, also due to Nekrasov, which resembles a unitary
random matrix integral. We estimate it using methods of random matrix
theory and an explicit formula from representation homology
\cite{BerestFelder2014}. The combinatorial formula for $Z$ above is
the sum of residues at certain poles of the integrand. As the choice
of integration cycle is essential for the estimate, we carefully prove
that the integral formula with the correct integration cycle is equal
to the combinatorial formula. This requires showing that certain
apparent residues actually cancel out.  The limit $\lambda\to 0$ to
the four dimensional theory is subtle and appears to require a
different approach; we hope to return to this problem in the future.

Another reason for the interest in the radius of convergence of the
power series $Z$ comes from conformal field theory in 2 dimensions
through the Alday--Gaiotto--Tachikawa (AGT) correspondence
\cite{AGT2009}. According to this correspondence, which was verified
in a number of cases, the instanton partition functions of
$\mathcal N=2$ supersymmetric gauge theories with suitable matter
fields are equal, up to a known scalar factor, to conformal blocks of
$W$-algebras. For example for $r=2$ and $\lambda=0$ they are related
to the four point conformal blocks of the Virasoro algebra is defined
as a matrix element of products of certain intertwining operators
(primary fields).  The parameters are the central charge
$c=13+6b^2+6b^{-2}$ with $\epsilon_1=b=\epsilon_2^{-1}$ and the
highest weights of four Virasoro representations, related to masses of
matter fields. A priori the four point conformal block is a formal
power series in the cross ratio $\mathfrak{q}$ of the four
points. Except in the special cases of degenerate representations,
occuring in minimal models at $c<1$, where they are solutions of
differential equations of hypergeometric type, giving full control on
the radius of convergence and analytic continuation, representation
theory does not seem to give information on the convergence of the
power series. This is particularly relevant for the unitary
representations of the Virasoro algebra with $c>1$, arising in
Liouville theory. We treat both the weakly coupled Liouville theory with
$c\geq 25$ and $\epsilon_1,\epsilon_2>0$ and the strongly coupled range
$1<c\leq 25$, with $\epsilon_1=\bar\epsilon_2$ on the unit circle.
The limiting case $c=1$ ($\epsilon_1=i=-\epsilon_2$) was considered for $\lambda=0$
in \cite{its2014}.
The case $\epsilon_1 = - \epsilon_2$ for $\lambda \neq 0$ was
 recently considered in \cite{Bershtein2017}.
In cases where the
AGT correspondence is understood, our estimates on the gauge theory
side imply the convergence of conformal blocks.  We make this explicit
in the mathematically well-understood case of the norm of Gaiotto
states for the $q$-deformed Virasoro algebra. Gaiotto states
\cite{Gaiotto2009} are Whittaker vectors in completions of Verma
modules of $W$-algebras and their deformations.  In the case of the
Virasoro algebra a Gaiotto state is a formal power series in the
eigenvalue of the generator $L_1$. The squared norm of a Gaiotto state
can be understood as suitable limits of a conformal block and
corresponds via the AGT correspondence to the Nekrasov partition
function of the pure Yang--Mills theory. As an application of our
result we prove that for a suitable range of parameters
the squared norm of the Gaiotto state for the
two-parameter deformation of the Virasoro algebra is a holomorphic
function of the eigenvalue with a convergence radius that converges to
infinity in the $\lambda\to 0$ limit.

  This agrees with the findings of Its et al.~\cite{its2014} and Bershtein et
  al.~\cite{Bershtein2017}. They consider the norm of the Gaiotto state in the
  case $\lambda = 0$ and $\lambda > 0$, respectively. Moreover, they require 
  $\epsilon_1 + \epsilon_2 = 0$, which is different from our setting. 
  In their respective setups, the Plancherel measure on partitions appears and
  allows for a direct estimate, which proves analyticity of the norm of the
  Gaiotto state on the whole complex plane.

The paper is organized as follows: in Section 2 we review the
mathematical definition of the Nekrasov partition function with matter
fields by introducing $M(r,n)$ as a special case of a Nakajima quiver
variety via the ADHM construction and explain the localization formula
leading to the combinatorial expression for $Z$. We then introduce the
integral representation of $Z$ and compute the large $n$ behaviour of
the coefficient of $\mathfrak{q}^n$ using methods of random matrix
theory and estimate the radius of convergence.  In Section 3 we discuss
how our method can be generalized to more general gauge theory,
including the $\mathcal N=2^*$ theory. In Section 4 we present the
application to the norm of Gaiotto states for the deformed Virasoro
algebra. In Section 5 we discuss some of the problems that are left
open.  The Appendix contains the proof of the integral representation
of the partition function.

\subsection*{Acknowledgements} We thank Mikhail Bershtein, Christoph Keller, Sara
Pa\-squet\-ti and J\"org Tesch\-ner for useful comments and
explanations. The authors were partially supported by the National
Centre of Competence in Research ``SwissMAP---The Mathematics of
Physics'' of the Swiss National Science Foundation.

\section{Nekrasov Partition Functions}
\label{sec:nekr-part-funct}
In this section, we define the $K$-theoretic Nekrasov
partition function as discussed in \cite{Nakajima2005}.
It is defined as a formal power series. Our aim is to show that for certain parameter ranges, the series converges.

\subsection{A Nakajima Quiver Variety}

We follow the exposition in \cite{Nakajima2003,Nakajima2005}.
Fix a positive integer $r$.  The Nekrasov partition function is a generating
functions for certain $K$-theory classes computed from a sequence of so-called Nakajima
quiver varieties $M(r,n)$, where $n=0,1,2,\dots$. They are constructed as
follows: Set $V = \C^n$ and $W = \C^r$. We define
\begin{align*}
  M(r,n) = \Big\{ & (B_1,B_2,i,j) \in \End(V) \times \End(V) \times \Hom(W, V)
  \times \Hom(V, W) : \\
  & [B_1,B_2] + ij = 0 \text{ and there does not exist a proper subspace $S
    \subset V$}
\\ &\text{such that } 
   B_1(S),B_2(S) \subset S \text{ and } i(W)
  \subset S\Big\} \Big/ GL(V),
\end{align*}
where an element $g \in GL(V)$ acts on $(B_1,B_2,i,j)$ via
\begin{align*}
  g \cdot (B_1,B_2,i,j) = (g B_1 g^{-1}, g B_2 g^{-1}, g i, j g^{-1}).
\end{align*}
The space $M(r,n)$ is a nonsingular algebraic variety of dimension $2nr$. 

Let $T \subset GL_r(\C)$ be the maximal torus consisting of diagonal
matrices. The group $\tT = T \times \C^* \times \C^*$ acts on $M(r,n)$. This
action is induced\footnote{Note our convention differs from the one used in
  \cite{Nakajima2003} by $e_\alpha \mapsto e_\alpha^{-1}$.} by
\begin{align*}
  (B_1,B_2,i,j) \cdot (t_1,t_2,e_1,\dots,e_r) = 
  (t_1 B_1, t_2 B_2, i e, t_1t_2 e^{-1} j),
\end{align*}
where $(t_1,t_2,e_1,\dots,e_r) \in \tT$ and $e$ is the diagonal matrix with
entries $e_1,\dots,e_r$.

We want to consider $\tT$-equivariant $K$-theory on $M(r,n)$. By the
localization principle, we first want a description of the fixed points of the
$\tT$-action. 
In order to describe the fixed points of this action, we use partitions. First
we fix our conventions:
We write $l(Y) = l$ for the length of the partition $Y = (Y(1),
\dots, Y(l))$ and $|Y|$ for its size $Y(1) + \dots + Y(l)$. We use the English
convention to draw the Young diagram corresponding to the partitions. For
example the partition $Y = (5,3,2)$ of size 10 has length $l(Y) = 3$ and
its Young diagram is given by \begin{align*}\yng(5,3,2).\end{align*} In the Young diagram,
the row index $x$ increases as we go south and the column index $y$
increases as we go east.  In the following we
identify the partition $Y$ with its Young diagram. In
particular, we write $Y = \{ (x,y) : 1 \leq x \leq l(Y), 1 \leq y \leq Y(x) \}$.

The fixed point set $M(r,n)^{\tT}$ is a discrete set whose points $I_{\vec{Y}}$
are
indexed by $r$-tuples $\vec{Y} = (Y_1, \dots, Y_r)$ of partitions with total
size $|\vec{Y}| := |Y_1| + \cdots + |Y_r| =n$. We denote the corresponding
inclusion maps as follows:
\begin{align*}
  \iota : M(r,n)^{\tT} &\to M(r,n),
  & \iota_{\vec{Y}} : \{ I_{\vec{Y}} \} &\to M(r,n), \qquad |\vec{Y}| = n.
\end{align*}

\subsection{Notations for K-Theory}
Let $t_i$ be the $\tT$-character given by
$(t_1,t_2,e_1,\dots,$ $e_r) \mapsto t_i$. Let $e_\alpha$ be the
$\tT$-character given by $(t_1,t_2,e_1,\dots,e_r) \mapsto
e_\alpha$. 
All the $\tT$-equivariant $K$-theory groups
$K^{\tT}(-)$ are modules for the ring
\begin{align*}
  R(\tT) := K^{\tT}( pt ) \cong \Z[ t_1^{\pm 1}, t_2^{\pm 1}, e_1^{\pm
    1}, \dots, e_r^{\pm 1}].
\end{align*}
Let $\mathcal{R}\cong \Q(t_1,t_2,e_1,\dots,e_r)$ denote its
quotient field. The maps $\iota$ and  $\iota_{\vec{Y}}$
all define pushforwards in K-theory. The Thomason localization
theorem in K-theory says that 
\begin{align*}
  \iota_* : K^{\tT}( M(r,n)^{\tT} ) \to K^{\tT}( M(r,n) )
\end{align*}
becomes an isomorphism after localization, i.e.~tensoring with $\mathcal{R}$. We
write 
\begin{align*}
  K_{\text{loc}}^{\tT}(-) = K^{\tT}(-)~\otimes_{R(\tT)}~\mathcal{R}.  
\end{align*}

\subsection{Tangent Spaces At Fixed Points}

We will need a
description of the tangent space $T_{\vec{Y}} M(r,n)$ at a fixed point
$I_{\vec{Y}}$ as a $\tT$-module. It is given by \cite[Theorem 2.11]{Nakajima2003}  
\begin{align}
  \label{eq:TY}
  T_{\vec{Y}} M(r,n) &= \sum_{\alpha, \beta = 1}^r
  N_{\alpha,\beta}(t_1,t_2) \in K^{\tT}(pt), \\ \nonumber
  N_{\alpha,\beta}(t_1,t_2)
  &= e_\alpha e_\beta^{-1} \bigg(
  \sum_{s \in Y_\alpha} t_1^{-l_{Y_\beta}(s)} t_2^{a_{Y_\alpha}(s)
    +1} +
  \sum_{t \in Y_\beta} t_1^{l_{Y_\alpha}(t)+1} t_2^{-a_{Y_\beta}(t)}
  \bigg). 
\end{align}
Here we use the arm length $a_Y(s) = Y(x) - y$ and leg length $l_Y(s) = Y^T(y) -
x$ of a box $s=(x,y) \in \Z^2$ with respect to the partition $Y$. The symbol
$Y^T$ denotes the transpose of the partition $Y$.  Note that box $s$ does not
have to belong to the partition $Y$, hence arm length and leg length can be
negative.

\subsection{The Tautological Bundle}
In order to include theories with massive matter, we introduce formal parameters
$\vec{b} = (b_1, \dots, b_s)$. The tautological bundle $\mathcal{V}$
is a vector bundle over $M(r,n)$ of rank $n$ whose fiber is $V$, see
e.g.~\cite{negut2015}.  Let $c_0, \dots, c_n \in K_{\text{loc}}^{\tT}(M(r,n))$
denote its Chern classes. Consider the
polynomial $b(z) = (z-b_1) \cdots (z-b_s)$. Construct the element 
\begin{align*}
  \kappa_n(b_1, \dots, b_s) := b(z_1) \cdots b(z_n) |_{\sigma_j = c_j} \in
  K_{\text{loc}}^{\tT}(M(r,n))[b_1, \dots, b_s]
\end{align*}
by expressing the symmetric polynomial $b(z_1) \cdots b(z_n)$ as a polynomial in
the elementary symmetric functions $\sigma_0, \dots, \sigma_n$ and replacing
each $\sigma_j$ by the Chern class $c_j$.

\subsection{K-theoretic Nekrasov Partition Function}
\label{sec:k-theoretic-version}

We now define the Nekrasov
partition function arising in the K-theory of the moduli space of instantons. Let $(\ee_1,\ee_2,a_1,\dots,a_r)$ be coordinates on the Lie
algebra of $\tT$ and $\lambda$ be a parameter such that $t_j = e^{\lambda
  \ee_j}$ and $e_\alpha = e^{\lambda a_\alpha}$. Write $b_m = e^{\lambda w_m}$. 
The Nekrasov partition function is defined as a formal power series
\begin{align}
  \label{eq:nekrasov}
  Z(\ee_1,\ee_2,\vec{a};\vec{w};\mathfrak{q},\lambda) = \sum_{n \geq 0} \big(
  \mathfrak{q} \lambda^{2r-s}  e^{-r \lambda ( \ee_1 + \ee_2 )/2} \big)^n
  Z_n(\ee_1,\ee_2, \vec{a};\vec{w}; \lambda),
\end{align}
with coefficients
\begin{align*}
  Z_n(\ee_1,\ee_2,\vec{a};\vec{w};\lambda) = &
  \sum_{i=0}^{2nr} (-1)^i \text{ch} \; H^i( M(r,n), \kappa_n(b_1, \dots, b_s) )
  \\ 
  =& 
  \sum_{|\vec{Y}|=n} (\iota_{*})^{-1}
  \Big( \kappa_n(b_1, \dots, b_s)\Big) \in \mathcal{R}[b_1,\dots,b_s].
\end{align*}
Here $\sum_{|\vec{Y}|=n}$ denotes the summation map
\begin{align*}
  K_{\text{loc}}^{\tT}(M(r,n)^{\tT}) [b_1,\dots,b_s] = \bigoplus_{|\vec{Y}|=n}
  \mathcal{R}[b_1,\dots,b_s] \to \mathcal{R}[b_1,\dots,b_s].
\end{align*}

The fixed point theorem says that after localization we have
\begin{align}
  \label{eq:kfixed}
  (\iota_*)^{-1} = \sum_{|\vec{Y}| = n} \frac{\iota_{\vec{Y}}^*}{ 
    \Lambda_{-1} T_{\vec{Y}} M(r,n) },
\end{align}
where $\Lambda_{-1}$ denotes the alternating sum of exterior powers. 
Combining equation \eqref{eq:kfixed} with the description 
(\ref{eq:TY}) of the tangent space at a fixed point to compute $\Lambda_{-1}
T_{\vec{Y}} M(r,n)$, one gets 
\begin{align}
  \label{eq:nekrasovcoeff}
  Z_n(\ee_1,\ee_2,\vec{a};\vec{w};\lambda)
  = \sum_{|\vec{Y}|=n}
  \frac
  {
    \iota_{\vec{Y}}^*\Big(\kappa_n(b_1, \dots, b_s)\Big)
  }
  {\prod_{\alpha,\beta=1}^r
    n^{\vec{Y}}_{\alpha,\beta}(\ee_1,\ee_2,\vec{a};\lambda)},
\end{align}
where
\begin{align*}
  n^{\vec{Y}}_{\alpha,\beta}(\ee_1,\ee_2,\vec{a};\lambda)
  =    
  & \prod_{s \in Y_\alpha} \bigg( 1 - e^{-\lambda \big( - l_{Y_\beta}(s) \ee_1 + (
    a_{Y_\alpha}(s) + 1 )\ee_2 + a_\alpha - a_\beta \big)} \bigg) \\
  &\prod_{t \in Y_\beta} \bigg( 1 - e^{-\lambda  \big( (l_{Y_\alpha}(t) + 1) \ee_1 
    - a_{Y_\beta}(t) \ee_2 + a_\alpha - a_\beta \big) } \bigg).
\end{align*}

For $s=0$, i.e.~$b(z)=1$, we have $\kappa_n(\es) = \mathcal{O}$,
the $K$-theory class corresponding to the structure sheaf on $M(r,n)$. This
corresponds to pure Yang Mills theory and the numerator in equation
\eqref{eq:nekrasovcoeff} equals one. For general $\vec{b} = (b_1, \dots, b_s)$
and a fixed point $I_{\vec{Y}}$, we have \cite[equation (2.27)]{negut2015}
\begin{align}
  \label{eq:pfixed}
  \iota_{\vec{Y}}^*\Big(\kappa_n(b_1, \dots, b_s)\Big) = 
  \prod_{\alpha=1}^r \prod_{(x,y) \in Y_\alpha} 
  \prod_{m=1}^s
  (e^{-\lambda( a_\alpha + (x-1) \ee_1 + (y-1) \ee_2) } -b_m).
\end{align}

Our aim is to estimate the general form of the coefficients
\eqref{eq:nekrasovcoeff} in order to prove convergence of the
$K$-theoretic partition function. To do so, we use an integral representation
for the coefficients $\ZZ_n(\ee_1,\ee_2,\vec{a},\vec{w};\lambda)$.

\subsection{Integral Representation and Estimate}

In this section we define an integral representation for the coefficients of the
Nekrasov partition function. We will use this representation to estimate the
coefficients. 

\subsubsection{Definition of the Integral}
\label{sec:definitionofintegral}
Let $q_1$ and $q_2$ be a pair of complex numbers  in
the open unit disk. Assume that either $q_1 = \overline{q_2}$ or $q_1, q_2
\in (0,1)$. Note that in either case $q_1 q_2 = |q_1 q_2| \in (0,1)$. Later these
numbers will be identified with exponential functions of $\ee_1,\ee_2$.

Fix $r \geq 1$. Let $\vec{u} = (u_1, \dots, u_r)$ be a vector of
complex numbers such that 
\begin{align*}
  |q_i| \max_{\alpha = 1, \dots, r} |u_\alpha| < \min_{\alpha=1,
    \dots, r} |u_\alpha| ,\qquad \forall\; i = 1,2.
\end{align*}
Let $\vec{p}= (p_1,\dots,p_s)$ be another vector of complex numbers.
The condition on $|u_\alpha|$ ensures that we can pick $\rho > 0$ with 
\begin{align}
  |u_\alpha| < \rho < |q_i|^{-1} 
  |u_\alpha| \qquad \forall\; \alpha = 1,\dots, r, \forall \; i =1,2.
  \label{eq:rhocond}
\end{align}
Let $C_\rho \subset \C$ be the set of complex numbers of modulus $\rho$. We define
\begin{align}
  \label{eq:zdef}
  \ZZ_n(\vec{u}; \vec{p}) =
  &
  \frac{1}{n!} \bigg( \frac{1-q_1q_2}{(1-q_1)(1-q_2)}\bigg)^n
  \int_{C^n_\rho} \prod_{j=1}^n \frac{dz_j}{2 \pi i z_j} 
  \;  \prod_{j=1}^n \prod_{m=1}^s (z_j-p_s)  \;
  \mathcal{I}(z_1, \dots, z_n;\vec{u}), 
\end{align}
where the integrand contains the symmetric function
\begin{align}
  \mathcal{I}(z_1,\dots,z_n;\vec{u}) = &
  \prod_{j=1}^n \prod_{\alpha=1}^r  \frac{ -u_\alpha z_j }
  {(z_j - u_\alpha) (q_1 q_2 z_j - u_\alpha)} 
  \label{eq:weight}
  \;
  \prod_{1 \leq j \neq k \leq n} 
  \frac{(z_j- z_k)(z_j - q_1q_2z_k)} {(z_j - q_1z_k)(z_j - q_2z_k)}.
\end{align}
In the following, we evaluate the coefficients $\ZZ_n(\vec{u};\vec{p})$ using
residue calculus.

\subsubsection{Evaluation of the Integral}
\label{sec:evaluationofintegral}

We evaluate the integral in $\ZZ_n(\vec{u};\vec{p})$ using residue calculus.
The residues are indexed by $r$-tuples $\vec{Y}$ of partitions $Y_\alpha$
with total size $|\vec{Y}| = |Y_1| + \cdots + |Y_r| = n$. We again identify
a partition with its Young diagram. For a box
$s=(x,y) \in Y_\alpha$ we define
\begin{align}
  \label{eq:residuedef}
  z^\alpha_s = z^\alpha_{x,y} = u_\alpha \; q_1^{x-1} q_2^{y-1}.
\end{align}
\begin{thm}
  \label{thm:residuecalculation}
  Under the assumption
  \begin{align}
    \label{eq:grid}
    u_\alpha u_\beta^{-1} &\neq  q_1^x q_2^y, &    
    \forall x, y & \in \{-n, \dots, n\},
    \quad 
    \forall \alpha \neq \beta \in \{1, \dots, r\}, 
    \\\nonumber
    q_1^x \neq q_2^{y+1}, \;\; q_1^{x+1} &\neq q_2^y,  &
    \forall x, y  & \in \{0, \dots, n-1\},
  \end{align}
  the value of $\ZZ_n(\vec{u};\vec{p})$ is given as a sum over $r$-tuples of partitions of
total size $n$ in two equivalent ways: 
  \begin{align}
    \label{eq:residueevaluation}
  \ZZ_n(\vec{u};\vec{p}) 
      &= 
      \sum_{|\vec{Y}|=n}
      \frac
      {
        \prod_{\alpha=1}^r \prod_{s \in Y_\alpha} \prod_{m=1}^s (z^\alpha_s-p_m)
      }
      {N^{\vec Y}_{\alpha,\beta}(q_1,q_2,\vec u)},
\\ \nonumber
 N^{\vec{Y}}_{\alpha,\beta}(q_1,q_2,\vec u)&=
      \prod_{s \in Y_\alpha}
      \Big(1 - \frac{u_\alpha}{u_\beta}
        q_1^{l_{Y_\alpha}(s) + 1} q_2^{-a_{Y_\beta}(s)} \Big)
      \prod_{t \in Y_\beta} 
      \Big(1 - \frac{u_\alpha}{u_\beta}
        q_1^{-l_{Y_\beta}(t)} q_2^{a_{Y_\alpha}(t)+1} \Big).
  \end{align}
    Alternatively,
    \begin{align}
      \label{eq:secondversion}
      \ZZ_n(\vec{u};\vec{p}) 
      &= 
      \sum_{|\vec{Y}|=n}
      \frac
      {
        \prod_{\alpha=1}^r \prod_{s \in Y_\alpha} \prod_{m=1}^s (z^\alpha_s-p_m)
      }
      {M_{\alpha,\beta}^{\vec Y}(q_1,q_2,\vec u)},
\\
 \nonumber
      M^{\vec{Y}}_{\alpha,\beta}(q_1,q_2,\vec u)&=
      \prod_{s \in Y_\alpha} 
      \Big(1 - \frac{u_\alpha}{u_\beta}
        q_1^{-l_{Y_\beta}(s)} q_2^{a_{Y_\alpha}(s)+1} \Big)
      \prod_{t \in Y_\beta}
      \Big(1 - \frac{u_\alpha}{u_\beta}
        q_1^{l_{Y_\alpha}(t) + 1} q_2^{-a_{Y_\beta}(t)} \Big).
    \end{align}
\end{thm}
The proof is of this theorem is technical. In the literature, several
arguments for the validity of this or similar formulae have been
given, see \cite{Nekrasov2002,FateevLitvinov2010,Hadasz2010,yanagida2010,negut2015}. They are based
on taking the iterated residues at $z^\alpha_{x,y}$,
$(x,y)\in Y_\alpha$, $\alpha=1,\dots,r$. However the integrand also
has further poles and it is a nontrivial fact, that we prove in this
paper, that the residues at those poles cancel. The proof of Theorem \ref{thm:residuecalculation}
is postponed to Appendix \ref{sec:residueproof}.
  \begin{rmk}
    \label{rmk:grid}
    The assumption \eqref{eq:grid} is necessary to ensure
  that all terms $\ZZ_{\vec{Y}}(\vec{u};\vec{p})$ are well-defined. If it is
  violated, some residues might not be simple residues anymore and consequently
  some $\ZZ_{\vec{Y}}(\vec{u};\vec{p})$ might be infinite. However,
  their sum $\ZZ_n(\vec{u};\vec{p})$ is still well-defined, as the integral in
  equation \eqref{eq:zdef} is.
\end{rmk}

  \subsubsection{Estimate for the Integral}
  \label{sec:estimateforintegral}

  In this section we apply potential theory to estimate the coefficients
  $\ZZ_n(\vec{u};\vec{p})$ in the integral form given by equation \eqref{eq:zdef} in the limit
  of large $n$.
  \begin{thm}
    \label{thm:limsupestimate}
    We have
    \begin{align*}
      \limsup_{n \to \infty} |\ZZ_n(\vec{u};\vec{p})|^{\frac{1}{n}} 
      \leq
      \prod_{m=1}^s \max\{ |p_m|, |u_1|, \dots,
      |u_r| \} .
    \end{align*}
  \end{thm}
  We prove this by comparing the growth of $\ZZ_n(\vec{u};\vec{p})$ to the growth of the
  coefficients 
  \begin{align}
    \label{eq:adef}
    a_n = & \frac{1}{n!}
    \bigg( \frac{1-q_1q_2}{(1-q_1)(1-q_2)} \bigg)^n
    \int_{C_1^n} \prod_{j=1}^n \frac{dz_j}{2\pi i z_j }
    \prod_{j \neq k} \frac{(z_j- z_k)(z_j - q_1q_2z_k)}{(z_j -
      q_1z_k)(z_j - q_2z_k)} .
  \end{align}
  We also want to introduce the language of potential theory.
  Let $\TT = \R/2 \pi \Z$ be the torus.  Define $f : \TT \to \R \cup \{ \infty \}$ by
  \begin{align}
    \label{eq:fdef}
    f(\theta) := - \log  \frac{ |e^{i\theta} - 1| | e^{i \theta} - q_1
      q_2|}{|e^{i \theta} - q_1|| e^{i \theta} - q_2|}.
  \end{align}
  For each $n \in \N$, we define a probability measure on $\TT^n$ by
  \begin{align*}
    P_n(\theta) = \frac{1}{Z_n} e^{- \sum_{j \neq k} f(\theta_k - \theta_j)},
  \end{align*}
  where $Z_n = \int_{\TT^n} d \theta e^{- \sum_{j \neq k} f(\theta_k - \theta_j)}.$
  Denote the associated expectation functionals by $\E_n[ - ]$.
  By changing variables $z_j = \rho e^{i \theta_j}$ in equation
  (\ref{eq:zdef}) and $z_j = e^{i \theta_j}$ in equation (\ref{eq:adef}), we get
  \begin{align*}
    \ZZ_n(\vec{u};\vec{p}) = & a_n
    \E_n\bigg[\prod_{j=1}^ng(\rho,\theta_j;\vec{u};\vec{p}) \bigg],
  \end{align*}
  where
  \begin{align*}
    g(\rho,\theta;\vec{u};\vec{p}) &= 
    \prod_{m=1}^s (\rho e^{i \theta} - p_m) \;
    \prod_{\alpha=1}^r 
    \frac{ - u_\alpha \rho e^{i\theta}}
    {(\rho e^{i\theta} - u_\alpha)
      (q_1 q_2 \rho e^{i\theta}- u_\alpha)}.
  \end{align*}
  We estimate $|\ZZ_n(\vec{u};\vec{p})|$ by taking the absolute value
  inside. We arrive at
  \begin{align}
    \label{eq:znestimate}
    \big| \ZZ_n(\vec{u};\vec{p}) \big|^{\frac{1}{n}} 
    \quad
    \leq 
    \quad
    \big| a_n \big|^{\frac{1}{n}}
    \;
    \exp\bigg( \frac{1}{n} \log \E_n\Big[ e^{\sum_j \log| g(\rho,\theta_j;\vec{u};p)
      |}\Big] \bigg). 
  \end{align}
  We estimate both factors on the right hand side separately. The first one is
  related to a known power series:
  \begin{lem}
    We have
    \label{lem:aestimate}
    \begin{align*}
      \limsup_{n \to \infty} | a_n |^{\frac{1}{n}} = 1.
    \end{align*}
  \end{lem}
  \begin{proof}
    From \cite{BerestFelder2014} we know that
    \begin{align*}
      \sum_{n \geq 0} a_n z^n = \exp\bigg( \sum_{n\geq 1}
      \frac{1-q_1^nq_2^n}{(1-q_1^n)(1-q_2^n)} \frac{z^n}{n} \bigg),
    \end{align*}
    as a formal power series. Since all coefficients are positive in this
    expansion and $q_i^n \to 0$ as $n \to \infty$ we
    see that the radius of convergence equals 1.
  \end{proof}

These estimates already give a non-quantitative convergence result.

\begin{lem}\label{lemma-0}
Assume that $\lambda>0$, $q_1,q_2$ in the unit disk, both real or complex
conjugate to each other. Then on each compact subset 
of the domain 
\[
\{(\vec u,\vec p)\in \mathbb C^{r+s}\,|\, |q_i|\max_{\alpha}|u_\alpha|<\min_\alpha|u_\alpha|,
\forall i=1,2\},
\]
the power series $\sum_{n=0}^\infty \ZZ_n(\vec u;\vec p)z^n$ converges uniformly
with a positive radius of convergence.
\end{lem}
To prove this lemma notice that on such compact subsets
$|g(\rho,\theta;\vec u;\vec p)|$ is uniformly bounded on the
integration circle $C_\rho$ for proper choice of $\rho$. Thus
the coefficients $\ZZ_n(\vec u;\vec p)$ are bounded by $\mathrm{const}^n$.

To get a quantitative estimate of the radius of convergence we need to
work harder.
We have the following result:

  \begin{thm}
    \label{thm:hlimitthm}
    Let $h$ be a continuous, real-valued function on the torus $\mathbb{T}$. We have
    \begin{align}
      \label{eq:hlimit}
      \frac{1}{n} \log \E_n\bigg[ e^{\sum_j h(\theta_j)}\bigg] \to
      \frac{1}{2\pi} \int_{\TT} h(\theta) d \theta \qquad (n \to \infty).
    \end{align}
  \end{thm}
  The proof of this theorem uses ideas from potential theory and is
  postponed to section \ref{sec:proofofhlimitthm}.  Now we can apply Lemma
  \ref{lem:aestimate} and Theorem \ref{thm:hlimitthm} in equation
  \eqref{eq:znestimate} to estimate 
  \begin{align*}
    \limsup_{n \to \infty} |\ZZ_n(\vec{u};\vec{p})|^{\frac{1}{n}} 
    \leq
    \exp \bigg( \frac{1}{2\pi} \int_{\TT} \log | g(\rho,\theta;\vec{u};\vec{p})
    | d \theta \bigg) .
  \end{align*}
  Using that condition (\ref{eq:rhocond}) says $\rho^{-1} 
  |u_\alpha | < 1$, but $ \rho^{-1} |q_1q_2|^{-1} |u_\alpha | > 1$,
  formula \ref{eq:ckgsigma}  from below implies
  \begin{align*}
    \frac{1}{2\pi} \int_{\TT} \log | g(\rho,\theta;\vec{u};p) | d \theta=
    \sum_{m=1}^s \max\{ \log |p_m|, \log \rho \} .
  \end{align*}
  We obtain
  \begin{align*}
    \limsup_{n \to \infty} |\ZZ_n(\vec{u};\vec{p})|^{\frac{1}{n}} 
    \leq
    \prod_{m=1}^s \max\{ |p_m|, \rho \} .
  \end{align*}
  According to condition (\ref{eq:rhocond}), the lower bound for $\rho$ is given
  by $\max_\alpha\{ |u_\alpha| \}$. We let $\rho$ tend to this
  bound, completing the proof of Theorem \ref{thm:limsupestimate}.

  \subsection{Potential Theory}
  \label{sec:proofofhlimitthm}

  In this section we prove Theorem \ref{thm:hlimitthm} using
  techniques adapted from \cite{johansson1998}.

  \subsubsection{Setup of Potential Theory}

  Let $M(\TT)$ be the set of all
  Borel probability measures on $\TT $ and $M_0(\TT)$ be the subset of
  all such measures which in addition do not contain point
  masses.

  The function $f : \TT \to \R \cup \{+\infty\}$ defined in equation
  \eqref{eq:fdef} is continuous, bounded
  from below and has a single pole at $\theta = 0$.
  Set
  \begin{align*}
    I[\mu] = \iint_{\theta \neq \phi} f(\theta - \phi) d \mu(\theta) d \mu(\phi),
  \end{align*}
  where $\mu \in M(\TT)$. Since $f$ is bounded from below, $I$ is
  bounded from below, too.
  Define
  \begin{align*}
    I_0 := \inf_{\mu \in M_0(\TT)} I[\mu] > - \infty.
  \end{align*}
  The aim of this section is to prove
  \begin{thm}
    \label{thm:I0unique}
    The normalized Lebesgue measure is the unique measure $\mu$ with
    $I[\mu] = I_0$. Moreover $I_0 = 0$.
  \end{thm}
  Define the function $g_\sigma(\theta) := \frac12 \log( 1 + \sigma^2 -
  2\sigma \cos \theta) = \log | e^{i\theta} - \sigma|$ for $\sigma > 0$.
  Its Fourier coefficients 
  \begin{align}
    \label{eq:ckgsigma}
    c_k(g_{\sigma}) = \frac{1}{2\pi} \int_{\TT} g_{\sigma}(\theta) e^{-ik\theta}
    d\theta = 
    \begin{cases}
      - \frac{1 } {2|k|}\min\{ \sigma^{|k|}, \sigma^{-|k|} \} , & \text{if $k \neq 0$,}
      \\
      \max\{ 0, \log \sigma \}, & \text{otherwise,}
    \end{cases}
  \end{align}
  are known. For $\sigma \neq 1$ see for example \cite{chen2002}. For the
  case $\sigma =1$, note that $g_\sigma$ is uniformly bounded from above and
  \begin{align*}
    g_\sigma(\theta) = \log( 1 - 2\sigma + \sigma^2 + 2\sigma -
    2\sigma \cos(\theta) ) \geq \log \sigma + g_1(\theta).
  \end{align*}
  By dominated convergence, it suffices to show that $g_1$ is
  integrable. The only pole is at $\theta=0$. By changing variables to
  $x = 2-2\cos\theta$ we have to consider $\log x \frac{1}{\sqrt{x}}$
  which is integrable as one can see from integration by parts.

  We write $q_j = |q_j| e^{i \tau_j}$. We have
  \begin{align*}
    c_k(f) = -c_k\big(g_1\big) - c_k\big(g_{|q_1q_2|}\big) + e^{-ik\tau_1}
    c_k\big(g_{|q_1|}\big) + e^{-ik \tau_2} c_k\big(g_{|q_2|}\big).
  \end{align*}
  Since $|q_i| < 1$, we obtain $c_0(f)=0$, so in particular
  \begin{lem}
    $I$ vanishes for the normalized Lebesgue measure.
  \end{lem}
  For $k\neq 0$ we get
  \begin{align*}
    c_k(f) = \frac{1}{2|k|} (1 + |q_1 q_2|^k - e^{-ik\tau_1} |q_1|^{|k|} -
    e^{-ik\tau_2} |q_2|^{|k|}).
  \end{align*}
  If $q_1 = \overline{q_2}$, this is bounded from below by $\frac{1}{2\abs{k}} ( 1
  - |q_1q_2|^{| k |/2} )^2 > 0$. If $q_1,q_2 \in (0,1)$ it equals $\frac{1}{2|k|}
  (1-|q_1|^{|k|})(1-|q_2|^{|k|}) > 0$. In either case we obtain
  \begin{lem}
    \label{lem:fourier}
    \label{lem:fl1}
    The Fourier coefficients $c_k(f)$ of $f \in L^1(\TT)$ satisfy
    $c_k(f) > 0$ if $k \neq 0$ and $c_0(f) = 0$. The Fourier series of
    $f$ converges everywhere except at $\theta = 0$.
  \end{lem}
  We use this to prove
  \begin{lem}
    We have $I_0 = 0$.
  \end{lem}
  \begin{proof}
    Since the normalized Lebesgue measure is in $M_0(\TT)$ we have $I_0
    \leq 0$. On the other hand, for any $\mu \in M_0(\TT)$, we have
    \begin{align*}
      I[\mu] & = \iint f(\theta - \phi) d \mu(\theta) d\mu (\phi)
      \\ 
      &= \sum_{k \neq 0} c_k(f) \iint e^{i (\theta - \phi)} d \mu(\theta)
      d \mu(\phi) \\ &= \sum_{k \neq 0} c_k(f) | c_k(\mu) |^2 \geq 0.
    \end{align*}
    Firstly, we have dropped the condition $\theta \neq \phi$ using the
    fact that $\mu$ does not contain point masses. Then we have applied
    Tonelli's theorem using the fact that $c_k(f) e^{i (\theta - \phi)}$
    is bounded from below.
  \end{proof}
  \begin{lem}
    The normalized Lebesgue measure is the unique measure $\mu \in
    M_0(\TT)$ for which $I_0 = I[\mu]$.
  \end{lem}
  \begin{proof}
    Let $\mu, \nu \in M_0(\TT)$ with $I[\mu] = I[\nu] = I_0$. Since
    $I[\mu]$ and $I[\nu]$ are both finite,
    \begin{align}
      \label{eq:exp}
      I[\mu-\nu] &= \iint f(\theta - \phi) d \mu(\theta) d\mu (\phi) +
      \iint f(\theta - \phi) d \nu(\theta) d\nu (\phi) \\\nonumber & - \iint f(\theta -
      \phi) d \mu(\theta) d\nu (\phi)
      - \iint f(\theta - \phi) d \nu(\theta) d\mu (\phi)  \in [-\infty, \infty)
    \end{align}
    is well-defined. Using Tonelli's theorem for each summand, we get
    \begin{align}
      \label{eq:sum}
      I[\mu-\nu] & = \sum_{k \neq 0} c_k(f) | c_k(\mu) - c_k(\nu)  |^2
      \geq 0.
    \end{align}
    In particular, $I[\mu-\nu]$ and all four terms in its expansion
    (\ref{eq:exp}) are finite.

    For $t \in [0,1]$ we have $\nu + t(\mu - \nu) \in M_0(\TT)$ and hence
    \begin{align*}
      0 = I_0 & \leq I[\nu + t (\mu-\nu)] \\
      & = I[\nu] + t \bigg( \iint f(\theta - \phi) d \nu(\theta) d(\mu-\nu) (\phi)\\ &+
      \iint f(\theta - \phi) d (\mu - \nu) (\theta) d\nu (\phi) \bigg) + t^2 I[\mu-\nu],
    \end{align*}
    where all terms are finite. The right hand side is a polynomial in
    $t$ which is nonnegative for $t \in [0,1]$ and vanishes at $t=0$ and
    $t=1$. We obtain $I[\mu-\nu] \leq 0$. From equation (\ref{eq:sum})
    we get $c_k(\mu) = c_k(\nu)$ for all $k \neq 0$, since $c_k(f) >
    0$ for $k \neq 0$. We obtain $\mu = \nu$ since $c_0(\mu) = c_0(\nu)$
    trivially because $\mu$ and $\nu$ are probability measures.
  \end{proof}
  The proof of Theorem \ref{thm:I0unique} is complete.

  \subsubsection{Application of Potential Theory}

  To a point $\theta \in \TT^n$ we associate the probability measure
  \begin{align*}
    \delta_\theta = \frac{1}{n} \sum_{j=1}^n \delta_{\theta_j},
  \end{align*}
  where on the right hand side we have a convex combination of ordinary
  Dirac measures.

  Set $\TT^n_0 = \{ \theta \in \TT^n : \theta_j \neq \theta_k (j \neq
  k)\}.$ For $\theta \in \TT^n_0$ we have $\theta_j \neq \theta_k
  \Leftrightarrow j \neq k$ and thus
  \begin{align*}
    n^2 I[\delta_{\theta}] = \sum_{j \neq k} f(\theta_j - \theta_k).
  \end{align*}

  The intuition behind Theorem \ref{thm:hlimitthm} is the following.
  The limit behavior for large $n$ of the quantity
  \begin{align*}
    \frac{1}{n} \log \E_n[e^{\sum_j
      h(\theta_j)}] = \frac{1}{n} \log 
    \frac{1}{Z_n} \int_{\TT^n} d\theta e^{\sum_j h(\theta_j)}
    e^{-n^2 I[\delta_\theta]}
  \end{align*}
  will be dominated by such $\theta \in \TT^n$, for which
  $I[\delta_\theta]$ is close to $I_0 = 0$. This is the content of Lemma
  \ref{lem:aest}. Those $\delta_\theta$ will
  then for large $n$ equidistribute to approximate the normalized
  Lebesgue measure yielding
  $\frac{1}{2\pi} \int_{\TT} h(\theta) d \theta$. This will be the
  content of Lemma \ref{lem:nulim} and the discussion afterwards.

  We use the notion of weak convergence for measures: We say a sequence of
  measures $(\mu_n)$ in $M(\TT)$ converges to a measure $\mu \in M(\TT)$ iff
  for all continuous and bounded functions $g$ on $\TT$ we have
  \begin{align*}
    \int_{\TT} g d \mu_n \to \int_{\TT} g d\mu.
  \end{align*}
  It is well-known that the space $M(\TT)$ with this notion of convergence is
  sequentially compact.

  For $\eta > 0$ define
  \begin{align*}
    A_{n, \eta} = \{ \theta \in \TT^n_0 :  I_0 \leq I[\delta_{\theta}] \leq I_0 + \eta\} = \{ \theta \in \TT^n :
    \sum_{j \neq k} f(\theta_k - \theta_j) \leq \eta n^2 \}.
  \end{align*}
  This set is compact. We have
  \begin{lem}
    \label{lem:aest}
    $0 \leq P_n[ \TT^n \setminus A_{n, \eta}] \leq e^{-\eta n^2}$.
  \end{lem}
  \begin{proof}
    We have
    \begin{align*}
      \int_{\TT^n \setminus A_{n, \eta}} d \theta e^{-\sum_{j \neq k}
        f(\theta_k - \theta_j)} \leq \int_{\TT^n \setminus A_{n, \eta}} d
      \theta e^{- \eta n^2} \leq (2 \pi)^n e^{-\eta n^2}.
    \end{align*}
    Let $\mu$ denote the normalized Lebesgue measure. By the Jensen inequality we
    have
    \begin{align*}
      \frac{Z_n}{(2\pi)^n} = \int_{\TT_0^n} \prod_j d\mu(\theta_j) e^{-
        \sum_{j \neq k} f(\theta_k - \theta_j)} \geq \exp \bigg(
      \int_{\TT_0^n} \prod_j d\mu(\theta_j) \big( - \sum_{j \neq k}
      f(\theta_k - \theta_j) \big) \bigg).
    \end{align*}
    Now
    \begin{align*}
      - \sum_{j \neq k}\int_{\TT_0^n} \prod_{j'} d\mu(\theta_{j'}) 
      f(\theta_k - \theta_j) &=-n(n-1) \iint f(\theta - \phi)
      d\mu(\theta) d \mu(\phi)
      \\ &= -n (n-1) I[\mu] = 0.
    \end{align*}
  \end{proof}
  Moreover, we have
  \begin{lem}
    \label{lem:nulim}
    If the measures $\nu_{n, \eta}$ are Dirac measures supported at
    $\tau^{n, \eta} \in A_{n, \eta}$ and $\nu_{n_k,\eta} \to \nu_\eta$
    is a convergent subsequence, we have $\nu_\eta \in M_0(\TT)$ and
    $I[\nu_\eta] \leq \eta$. If $\nu_{\eta_k} \to \nu$ is a convergent
    sequence with $\eta_k \to 0$, the limit $\nu$ has to be the
    normalized Lebesgue measure.
  \end{lem}
  \begin{proof}
    Let $L \in \R$ and separate the diagonal part:
    \begin{align*}
      \eta \geq I[\nu_{n, \eta}] & = \iint_{\theta \neq \phi} f(\theta - \phi) d \nu_{n,
        \eta}(\theta) d \nu_{n, \eta}(\phi) 
      \\ &\geq 
      \iint_{\theta \neq \phi} \min\{f(\theta - \phi),L \} d \nu_{n,
        \eta}(\theta) d \nu_{n, \eta}(\phi)\\
      & = \frac{1}{n^2} \sum_{j \neq k} \min\{f(\tau^{n,\eta}_j -
      \tau^{n,\eta}_k),L \} 
      \\ &=
      \iint \min\{f(\theta - \phi),L \} d \nu_{n,
        \eta}(\theta) d \nu_{n, \eta}(\phi)
      - \frac{L}{n}.
    \end{align*}
    Let $\epsilon > 0$. Using the Weierstrass approximation theorem,
    pick a polynomial
    $p(\theta,\phi)$ which uniformly approximates the last integrand up
    to an error of $\epsilon$. For $n = n_k$ we get
    \begin{align*}
      \eta \geq       \iint p(\theta,\phi) d \nu_{n_k,
        \eta}(\theta) d \nu_{n_k, \eta}(\phi)- \epsilon
      - \frac{L}{n_k}. 
    \end{align*}
    Send $k \to \infty$ to get
    \begin{align*}
      \eta \geq  \iint p(\theta,\phi) d \nu_{\eta}(\theta) d
      \nu_{\eta}(\phi) -\epsilon \geq \iint \min\{f(\theta - \phi),L \} d \nu_\eta(\theta) d
      \nu_\eta(\phi) - 2 \varepsilon.
    \end{align*}
    Now send $\epsilon \to 0$ and let $L \to \infty$. By monotonicity
    the limit can pass to the integrand. We obtain
    $\nu_\eta \in M_0(\TT)$ and
    \begin{align*}
      \eta \geq \iint f(\theta - \phi)d \nu_\eta(\theta) d\nu_\eta(\phi)
      =\iint_{\phi \neq \theta} f(\theta - \phi)d \nu_\eta(\theta) d\nu_\eta(\phi) = I[\nu_\eta].
    \end{align*}
    Now let $\nu_{\eta_k} \to \nu$ be a convergent sequence with $\eta_k
    \to 0$. Again fix $L \in \R$ and estimate, using $\nu_{\eta_k} \in M_0(\TT)$,
    \begin{align*}
      \eta_k \geq I[\nu_{\eta_k}] = \iint
      f(\theta-\phi) d\nu_{\eta}(\theta) d \nu_{\eta}(\phi) 
      \geq \iint
      \min\{ L, f(\theta-\phi)  \} d\nu_{\eta}(\theta) d \nu_{\eta}(\phi). 
    \end{align*}
    Let $\epsilon > 0$. Using Weierstrass we get in the limit $k \to
    \infty$
    \begin{align*}
      0 \geq \iint
      \min\{ L, f(\theta-\phi)  \} d\nu(\theta) d \nu(\phi)  - 2 \epsilon.
    \end{align*}
    Again, we let $\epsilon \to 0$ and $L \to \infty$ to get $\nu \in
    M_0(\TT)$ and $0 \geq I[\nu]$. The claim follows.
  \end{proof}
  Now we can prove Theorem \ref{thm:hlimitthm}. Let $h$ be a continuous function on
  the torus. Fix $\eta > 0$.  By Lemma \ref{lem:aest}, we have
  \begin{align*}
    \limsup_{n \to \infty} \frac{1}{n}\log \E_n[ e^{\sum_j h(\theta_j)}]
    =
    \limsup_{n \to \infty} \frac{1}{n}\log \int_{A_{n,\eta}} d \theta  P_n(\theta) e^{\sum_j h(\theta_j)},
  \end{align*}
  and analogously for $\liminf$.
  Let the continuous, real-valued function $e^{\sum_j h(\theta_j)}$ on the
  compact set $A_{n, \eta}$ attain its maximum at $\tau^{n,\eta} \in A_{n,
    \eta}$ and its minimum at $\sigma^{n,\eta} \in A_{n,\eta}$. Denote
  Dirac measures by $\nu_{n, \eta}$ and $\lambda_{n,\eta}$. We have
  \begin{align*}
    \frac{1}{n}\log \int_{A_{n,\eta}} d \theta  P_n(\theta) e^{\sum_j
      h(\theta_j)} \leq \frac{1}{n} \log P_n[A_{n,\eta}] e^{\sum_j
      h(\tau^{n,\eta}_j)}  \leq \int_{\TT} d \nu_{n,\eta}(\theta)  h(\theta).
  \end{align*}
  Similarly,
  \begin{align*}
    \frac{1}{n}\log \int_{A_{n,\eta}} d \theta  P_n(\theta) e^{\sum_j
      h(\theta_j)} \geq \frac{1}{n} \log P_n[A_{n,\eta}] e^{\sum_j
      h(\sigma^{n,\eta}_j)}  \geq \int_{\TT} d \lambda_{n,\eta}(\theta)
    h(\theta) + O( \frac{1}{n} )
  \end{align*}
  since by Lemma \ref{lem:aest}, $P_n[A_{n,\eta}] \geq 1 - e^{-\eta n^2}
  \to 1$. Let $n_k$ define a subsequence with
  \begin{align*}
    \limsup_{n \to
      \infty} \frac{1}{n}\log \int_{A_{n,\eta}} d \theta  P_n(\theta)
    e^{\sum_j h(\theta_j)} =
    \lim_{k \to \infty} \frac{1}{n_k}\log \int_{A_{n_k,\eta}} d \theta
    P_{n_k}(\theta) e^{\sum_j h(\theta_j)}.
  \end{align*}
  Let $m_k$ define a subsequence that realizes the corresponding
  $\liminf$. By passing to respective subsequences, we can suppose that
  $\nu_{n_k, \eta} \to \nu_{\eta}$ and
  $\lambda_{m_k,\eta} \to \lambda_{\eta}$ for some Borel probability
  measures $\lambda_{\eta},\nu_{\eta}$. We obtain
  \begin{align*}
    \int_{\TT} h(\theta) d\lambda_{\eta}(\theta) &\leq \liminf_{n \to
      \infty}
    \frac{1}{n}\log \E_n[ e^{\sum_j h(\theta_j)}] \\ &\leq
    \limsup_{n \to \infty} \frac{1}{n}\log \E_n[ e^{\sum_j
      h(\theta_j)}] \leq \int_{\TT} h(\theta) d \nu_{\eta}(\theta).
  \end{align*}
  The parameter $\eta > 0$ is arbitrary.
  Now let $\eta_k \to 0$ define a subsequence $\nu_{\eta} \to \nu$ and
  $\eta'_k \to 0$ define a subsequence $\lambda_{\eta} \to \lambda$. By
  Lemma \ref{lem:nulim}, $\nu = \lambda$ is the normalized
  Lebesgue measure. We obtain
  \begin{align*}
    \frac{1}{2\pi} \int_{\TT} h(\theta) d \theta &\leq \liminf_{n \to
      \infty}
    \frac{1}{n}\log \E_n[ e^{\sum_j h(\theta_j)}] \\ &\leq
    \limsup_{n \to \infty} \frac{1}{n}\log \E_n[ e^{\sum_j
      h(\theta_j)}] \leq \frac{1}{2 \pi} \int_0^{2 \pi} h(\theta) d \theta.
  \end{align*}
  The proof of Theorem \ref{thm:hlimitthm} is complete. 
  \begin{rmk}
    A physical interpretation of Theorem \ref{thm:hlimitthm} goes as follows:
    The points of $\TT^n$ are the coordinates of $n$ particles on the torus
    $\TT$ that interact via the two body potential $f(\theta_j - \theta_k)$. The
    integral over $\TT^n$ on the left hand side of equation \eqref{eq:hlimit} is
    dominated by particle configurations which minimize the potential energy of
    the system. The even function $f(\theta)$ has a pole at $\theta = 0$ and one
    minimum $\theta_0$ in $(0, \pi)$. Hence the particles experience a strong
    repulsive force once they get close to each other and have a preferred
    distance $\theta_0$ from each other. For large $n$ they cannot all stay in
    their preferred distance since the torus is compact. Hence the repulsive part
    dominates and in the limit of large $n$ the particles equidistribute. The
    evaluation of the left hand side of equation \eqref{eq:hlimit} on those
    equidistributed points on $\TT$ defines a Riemann sum approximating the
    integral on the right hand side of equation \eqref{eq:hlimit}. 
  \end{rmk}

  \subsection{Identification of Coefficients and Conclusion}
  \label{sec:identificationofcoefficients}
  We evaluate the formal parameters of section \ref{sec:k-theoretic-version} at the
  complex parameters of section \ref{sec:definitionofintegral}. We set
  \begin{align*}
    t_i^{-1} = e^{-\lambda \ee_i} = q_i,  \qquad 
    e_\alpha^{-1} =  e^{-\lambda a_\alpha} =
    u_\alpha, \qquad  b_m = e^{\lambda w_m} = p_m.
  \end{align*}
  Under the assumption $\lambda > 0$, the conditions stated at the
  beginning of section \ref{sec:definitionofintegral} are satisfied provided
  we assume $\Re \ee_i > 0$ and $\big(\ee_1 = \overline{\ee_2}$ or $\ee_1, \ee_2
  \in \R \big)$ and
  \begin{align*}
    \max_\alpha
    \Re(a_\alpha) - \min_\alpha \Re(a_\alpha) < \Re \ee_i , \quad i = 1,2.
  \end{align*}
  Hence $\ZZ_n(\vec{u};\vec{p})$ is well-defined. To ensure that each simple
  residue $\ZZ_{\vec{Y}}(\vec{u};\vec{p})$ is well-defined, we have to require
  conditions \eqref{eq:grid}, which translate to
  \begin{align}
    \label{eq:grid1}
    a_\alpha - a_\beta  \not\equiv x \ee_1 + y \ee_2, 
    \qquad
    \forall \alpha \neq \beta \in \{1, \dots, r\}
    \quad
    \forall x, y & \in \{-n, \dots, n\},
    \\\nonumber 
     x \ee_1 \not\equiv (y+1) \ee_2,\; (x+1)\ee_1  \not\equiv y \ee_2, 
     \qquad
     \qquad
     \qquad
    \forall x, y  & \in \{0, \dots, n-1\},
  \end{align}
  where the inequalities are modulo $\frac{2\pi i}{\lambda} \Z$.
  Under these conditions, Theorem \ref{thm:residuecalculation} implies
  $Z_n(\ee_1,\ee_2, \vec{a};\vec{w}; \lambda) = \ZZ_n(\vec{u};\vec{p})$. 
  By Remark \ref{rmk:grid}, the sum $Z_n(\ee_1, \ee_2, \vec{a}; \vec{w})$ is
  well-defined even if we drop condition \eqref{eq:grid1}. 
  The representation of the coefficients as integrals allowed us to apply the
  potential theory in section \ref{sec:estimateforintegral} to estimate their
  growth in Theorem \ref{thm:limsupestimate}. We obtain

\begin{thm}
    \label{thm:nekrasovconvergence}
    Let $\lambda>0$ and $\ee_1$ and $\ee_2$ be a pair of complex
    numbers with positive real part.  Assume either they are complex
    conjugate or both real. 
The $K$-theoretic Nekrasov partition function $Z(\ee_1, \ee_2, \vec{a};
    \vec{w},
    \mathfrak{q}, \lambda)$ given by equation \eqref{eq:nekrasov} is an analytic
    function of $\mathfrak{q}, \vec a, \vec w$ in the domain
    \begin{align*}
\\
\vec w\in\mathbb C^s,  \quad      \max_\alpha \Re(a_\alpha) - \min_\alpha \Re(a_\alpha) 
      < \Re \ee_i , \quad i = 1,2,
\\
    |\mathfrak{q}| < \lambda^{s-2r} 
      e^{ \lambda \big( 
        r \frac{\ee_1 + \ee_2}{2}  + \sum_{m=1}^s
        \min\{ -\Re w_m, \Re a_1, \dots, \Re a_r \} \big)
      }.
    \end{align*}
  \end{thm}
  There remains to prove the analyticity in $\vec a,\vec w$. By
  Lemma \ref{lemma-0}, we know that in a neighbourhood $U$ of each
  $(\vec a,\vec w)$ the series defining the partition functions
  converges to an analytic function for $|\mathfrak q|<r$ for some
  small positive radius $r=r(U)$. On the other hand, Theorem
  \ref{thm:limsupestimate} tells us that we have
  convergence of the power series for the indicated range of $\mathfrak q$ and
  fixed $\vec a$, $\vec w$ in the domain. By Hartogs' lemma
  (\cite{BochnerMartin}, Theorem 2, p.~139) the series converges to an
  analytic function on the whole domain.

  \section{Generalization}
  \label{sec:general}
  Our technique can be generalized to different parameter ranges and similar
  types of gauge theory.

  \subsection{Different Parameter Ranges}
  Assume $\vec{p} = \es$. Under the inversion 
  \begin{align*}
    (q_1, q_2, u_1, \dots, u_r) \mapsto (q_1^{-1}, q_2^{-1}, u_1^{-1}, \dots,
    u_r^{-1} )
  \end{align*}
  the coefficient $\ZZ_n(\vec{u};\vec{p} = \es)$ gets multiplied by $(q_1
  q_2)^{-nr}$. This can directly be seen from equation \eqref{eq:secondversion}.
  Using this observation, bounds for $|q_1|, |q_2| > 1$ can be obtained as well.

  \subsection{Similar Types of Gauge Theory}
  Our technique readily generalizes to partition functions of gauge theories where
  the weight factor in line \eqref{eq:weight} in coordinates $z_j = \rho e^{i
    \theta_j}$ is of the form
  \begin{align*}
    \exp\bigg( - \sum_{j \neq k} f(\theta_k - \theta_k) \bigg),
  \end{align*}
  where $f : \TT \to \R \cup \{ + \infty\}$ is continuous, bounded from below, has
  a single pole at $\theta = 0$ and its Fourier coefficients $c_k(f)$ satisfy
  $c_k(f) > 0$ for $k \neq 0$ and $c_0(f) =0$.

  An example would be $\mathcal{N} =
  2^*$ supersymmetric gauge theory, i.e.~with massive matter in the adjoint representation
  \cite[equation (3.25)]{Nekrasov2002}. One has to consider a different polynomial
  in the numerator of the integrand in \eqref{eq:zdef} and multiply the
  weight in line \eqref{eq:weight} by
  \begin{align*}
    \prod_{1 \leq j < k \leq n} 
    \frac
    {
      (z_j - q_1 m^{-1} z_k)(z_j - q_2 m^{-1} z_k) 
      (z_j - q_1^{-1} m z_k)(z_j - q_2^{-1} m z_k)}
    {
      (z_j- m z_k)(z_j - m^{-1} z_k )
      (z_j - q_1q_2 m^{-1} z_k)(z_j - q_1^{-1} q_2^{-1} m z_k)}.
  \end{align*}

  \section{Specialization: Norm of Deformed Gaiotto States}
  \label{sec:gaiotto}

  In this section, we briefly discuss an application of our results to conformal
  field theory, namely the finiteness of the norm of deformed Gaiotto
  states.

  \subsection{Definition}
  In this introduction, let $q, t$ denote two generic complex
  parameters. We follow the exposition of \cite{Awata2009}.  The
  deformed Virasoro algebra $Vir_{q,t}$ is defined as the
  associative algebra topologically generated by $T_n, n \in \Z$ with the defining
  relations
  \begin{align*}
    [T_n, T_m] = &-\sum_{l=1}^{\infty}  r_l( T_{n-l}T_{m+l} - T_{m-l}
    T_{n+l}) \\ &- \frac{(1-q)(1-t^{-1})}{1-qt^{-1}} ( q^nt^{-n} - q^{-n}
    t^{n}) \delta_{m+n,0}.
  \end{align*}
  The coefficients $r_l$ are determined by the expansion
  \begin{align*}
    r(x) = \sum_{l \geq 0} r_l x^l = \exp\bigg( \sum_{n \geq 1} \frac{
      (1-q^n)(1-t^{-n})}{ 1 + q^nt^{-n}} \frac{x^n}{n}\bigg).
  \end{align*}

  For $h \in \C$, the Verma module $M_h$ is generated by a vector
  $| h \ra$ with $T_0 |h\ra = h |h \ra$ and $T_n | h \ra = 0$ if
  $n \geq 1$. The operator $T_0$ defines a grading $M_h = \bigoplus_{n \geq 0}
  M_{h, n}$ where $M_{h,n}$ is the eigenspace of $T_0$ corresponding to the
  eigenvalue $h+n$. As usual, each $M_{h,n}$ has a basis 
  \begin{align*}
    T_\lambda |h\ra := T_{-\lambda(1)} \cdots T_{-\lambda(l)} | h \ra
  \end{align*}
  indexed by
  partitions $\lambda = (\lambda(1), \dots, \lambda(l))$ of size $|\lambda|=n$.
  The Shapovalov form $S : M_h \otimes M_h \to \C$ is characterized by $S(| h \ra,
  | h \ra) = 1$ and $S(T_n x, y) = S(x, T_{-n}y)$
  for $x,y \in M_h$ and $n \in \Z$. The decomposition $M_h = \bigoplus_{n \geq 0}
  M_{h,n}$ is orthogonal with respect to the Shapovalov form. In particular, the
  so-called Kac matrix with entries $S_{\lambda,\mu} := S( T_{\lambda} |h \ra, T_{\mu} | h \ra)$ is
  block diagonal with finite blocks $(S^{(n)}_{\lambda, \mu})_{|\lambda|, |\mu| =
    1, \dots, n}$.

  A deformed Gaiotto state
  is a formal power series $|G\ra = \sum_{n \geq 0} \xi^n | G_n \ra$ whose 
  coefficients $| G_n \ra \in M_{h,n}$ satisfy $T_1 |G_n \ra = |G_{n-1} \ra$ with
  $| G_{-1} \ra := 0$ and $T_n |G_m \ra = 0$ for $n \geq 2$ and all $m$. In terms
  of the expansion $|G_n \ra = \sum_{|\lambda| = n} g^{(n)}_{\lambda} T_{\lambda}
  | h
  \ra$ these conditions read
  \begin{align*}
    \sum_{|\lambda| = n} g^{(n)}_{\lambda} S^{(n)}_{\lambda,\mu} = \delta_{\mu
      (1^n)}, \qquad n = 0, 1, \dots,
  \end{align*}
  where $(1^n)$ is the partition $(1, \dots, 1)$ of size $n$.

  From now on, we assume $q, t \neq 1$. 
  The zeros of the determinant of the $n$-th block $S^{(n)}$ of the Kac matrix are
  located at \cite[equation (2.4)]{yanagida2014} 
  \begin{align}
    \label{eq:kac}
    h = \pm ( t^{r/2} q^{-s/2} + t^{-r/2} q^{s/2}), \qquad r,s \geq 1, rs = n.
  \end{align}
  Outside of these sets, the Kac matrix is invertible, hence the Gaiotto state
  exists uniquely and its norm with
  respect to the Shapovalov form is given by the formal power series
  \begin{align*}
    S(|G\ra, |G\ra) = \sum_{ n \geq 0} \xi^{2n}
    (S^{(n)})^{-1}_{(1^n),(1^n)}.
  \end{align*}

  Let $Q \in \C$ with $h = Q^{\frac12} +
  Q^{-\frac12}$. From the AGT relation (\cite{Awata2009}, \cite{yanagida2010} and
  \cite{yanagida2014}) we obtain the formal expansion as a Nekrasov
  partition function
  \begin{align}
    \label{eq:formalexp}
    S(|G\ra, |G\ra) = \sum_{n \geq 0} \xi^{2n} q^{-n} t^n Z_n(q,t,Q),
  \end{align}
  where the coefficients are sums over pairs $(\nu,\mu)$ of partitions:
  \begin{align*}
    Z_n(q,t,Q) &=  \sum_{|\nu|+|\mu|=n} \frac{1}{N_{\nu,\mu}(Q)
      N_{\mu,\nu}(Q^{-1}) N_{\nu,\nu}(1) N_{\mu,\mu}(1)}, \\
    N_{\nu,\mu}(Q) & = 
    \prod_{t \in \mu} (1-Q q^{a_{\nu}(t)} t^{l_\mu(t) +1})
    \prod_{s \in \nu} (1-Q q^{-a_{\mu}(s)-1} t^{-l_\nu(s)})
  \end{align*}

  \subsection{Finiteness of Norm}
  We want to recover the coefficients defined in equation \eqref{eq:zdef}.
  We set $q_1 = t$, $q_2 = q^{-1}$, $r = 2$ and $u_1 = u_2^{-1} = Q^{\frac12}$.
  The conditions stated at the beginning of section
  \ref{sec:definitionofintegral} say $|q| >1$, $|t|<1$ and $\big(t \overline{q}
  = 1$ or $t,q > 0\big)$ and 
  \begin{align*}
  |t| \max\{ |Q|, |Q|^{-1} \} < \min\{ |Q|, |Q|^{-1} \}, \qquad
    \max\{ |Q|, |Q|^{-1} \} < |q| \min\{ |Q|, |Q|^{-1} \}.
  \end{align*}
  The condition \eqref{eq:grid} on the well-definedness of all simple residues
  translates to
  \begin{align*}
    Q &\notin \{ q^x t^y : x, y \in \Z \}, \qquad  \text{and} \qquad
    q^x \neq t^y  \; \forall (x, y) \in \Z^2 \setminus \{(0,0)\}.
  \end{align*}
  We extend the condition on $Q$ on the left to $-Q$, thus excluding all poles
  of the Kac determinant in equation \eqref{eq:kac}. The condition on the right
  we drop by remark \eqref{rmk:grid}.
  From equation \eqref{eq:residueevaluation} in Theorem \ref{thm:residuecalculation}, we then recover 
  \begin{align}
    \label{eq:gaiottoseries}
    S(|G \ra, |G \ra) = \sum_{n \geq 0} z^n \ZZ_n(\vec{u};p=\es)
  \end{align}
  for $z = \xi^2 t q^{-1}$.  
 Hence, Theorem \ref{thm:limsupestimate}
  implies that series \eqref{eq:gaiottoseries} converges for $|z| < 1$.
  We obtain
  \begin{thm}
    \label{thm:gaiottothm}
    Let $q$ and $t$ be a pair of complex numbers of with $|t|<1$ and $|q|>1$
    Suppose either $t\overline{q} = 1$ or $t,q > 0$. Let $Q^{\frac12}$ be a complex number such that
    \begin{align*}
       Q &\notin \{ q^x t^y : x, y \in \Z \},
      & |q| \max\{ |Q|, |Q|^{-1} \} &< \min\{ |Q|, |Q|^{-1} \}, \\
     - Q &\notin \{ q^x t^y : x, y \in \Z \},
      &
      \max\{ |Q|, |Q|^{-1} \} &< |t| \min\{ |Q|, |Q|^{-1} \}.
    \end{align*}
    Set $h = Q^{\frac12} + Q^{-\frac12}$.
    The deformed Gaiotto state $|G\ra$
    with formal parameter $\xi \in \C$ exists for the Verma module $M_h$ 
    for the deformed Virasoro algebra $Vir_{q,t}$, and its norm $S(|G \ra, |G \ra)$ 
    is an analytic function in the variable $\xi$ for
    \begin{align*}
      |\xi| <  |t|^{-\frac12} |q|^{\frac12}.
    \end{align*}
  \end{thm}

  \section{Open Problems}
  \label{sec:open}

  Here, we briefly describe two formal power series, to which we intended to apply
  our results. Their coefficients are certain limits of the coefficients we studied
  earlier. Due to the lack of uniformity in our estimates, we were not able to
  prove convergence of those power series.

  \subsection{Homological Version of Nekrasov Partition Function}
  \label{sec:homological-version}

  The $K$-theoretic Nekrasov partition function defined in
  section \ref{sec:nekr-part-funct} has a counterpart in which $K$-theory groups
  are replaced by $\tT$-equivariant Borel--Moore homology groups
  $H_*^{\tT}(-)$. All $H_*^{\tT}(-)$ are modules for the ring
  \begin{align*}
    S(\tT) := H_*^{\tT}( pt ) \cong \Z[ \ee_1, \ee_2, a_1, \dots, a_r].
  \end{align*}
  Let $\mathcal{S}\cong \Q(\ee_1,\ee_2,a_1,\dots,a_r)$ denote its
  quotient field. The maps $\iota$ and $\iota_{\vec{Y}}$,  defined in
  section \ref{sec:nekr-part-funct},
  all define pushforwards in equivariant homology. The localization
  theorem in equivariant homology says that $\iota$ becomes an
  isomorphism after tensoring with $\mathcal{S}$. 

  The fixed point theorem says that after tensoring with $\mathcal{S}$
  we have
  \begin{align}
    \label{eq:hfixed}
    (\iota_*)^{-1} = \sum_{|\vec{Y}| = n} \frac{\iota_{\vec{Y}}^*}{ e(
      T_{\vec{Y}} M(r,n)) },
  \end{align}
  where $e(-)$ denotes the Euler class.
  The homological Nekrasov partition function for pure Yang
  Mills theory is defined as a formal power series
  \begin{align*}
    Z(\ee_1,\ee_2,\vec{a};\mathfrak{q}) = \sum_{n \geq 0} \mathfrak{q}^n
    Z_n(\ee_1,\ee_2, \vec{a}),
  \end{align*}
  with coefficients
  $Z_n(\ee_1,\ee_2,\vec{a}) = \sum_{|\vec{Y}|=n} (\iota_{*})^{-1} [M(r,n)]$ where
  $[M(r,n)]$ denotes the fundamental class of $M(r,n)$. Using the fixed point formula
  (\ref{eq:hfixed}) and equation (\ref{eq:TY}) to compute the Euler
  classes gives
  \begin{align}
    Z_n(\ee_1,\ee_2,\vec{a})
    = \sum_{|\vec{Y}|=n}
    \frac{1}{\prod_{\alpha,\beta}
      n^{\vec{Y}}_{\alpha,\beta}(\ee_1,\ee_2,\vec{a})},  
  \end{align}
  where
  \begin{align*}
    n^{\vec{Y}}_{\alpha,\beta}(\ee_1,\ee_2,\vec{a})
    =    
    & \prod_{s \in Y_\alpha} \big( - l_{Y_\beta}(s) \ee_1 + (
    a_{Y_\alpha}(s) + 1 )\ee_2 + a_\alpha - a_\beta \big) \\
    &\prod_{t \in Y_\beta} \big( (l_{Y_\alpha}(t) + 1) \ee_1 
    - a_{Y_\beta}(t) \ee_2 + a_\alpha - a_\beta \big).
  \end{align*}

  In particular, the coefficients of the homological Nekrasov partition
  function are limits of the coefficients of the K-theoretic partition function,
  as defined in equation \eqref{eq:nekrasovcoeff}:
  \begin{align}
    \label{eq:neklimit}
    \lambda^{2 r n} Z_n(\ee_1,\ee_2,\vec{a};\vec{w}=\es;\lambda) \to   Z_n(\ee_1,\ee_2,\vec{a}) \qquad
    (\lambda \to 0).
  \end{align}

  Equation \eqref{eq:neklimit} suggests that, under the same conditions on $\ee_1,
  \ee_2, a_1, \dots, a_r$ as stated in Theorem \ref{thm:nekrasovconvergence}, the
  homological Nekrasov partition function $Z(\ee_1,\ee_2,\vec{a};\mathfrak{q})$
  converges for all $\mathfrak{q} \in \C$.
  This question remains unanswered: Neither is the limit \eqref{eq:neklimit}
  uniform in $n$, nor is the estimate for the radius of convergence uniform in
  $\lambda$. 

  \begin{rmk}
    In the special case $\epsilon_1 + \epsilon_2 = 0$, one can directly estimate the
    coefficients $Z_n$ 
to prove convergence of the expansion
    $Z(\ee_1,\ee_2,\vec{a};\mathfrak{q}) = \sum_{n \geq 0} \mathfrak{q}^n
    Z_n(\ee_1,\ee_2,\vec{a})$: The diagonal term
    $\alpha=\beta$ in $\prod_{\alpha,\beta=1}^n
    n^{\vec{Y}}_{\alpha,\beta}(\ee_1,\ee_2,\vec{a})$ is a product over all
    hook lengths occurring in the Young diagram $Y_\alpha$. By the hook length
    formula we obtain $\frac{1}{|Y_\alpha|!}$ times the Plancherel measure of the
    Young diagram $Y_\alpha$, which allows one to estimate the expansion for all
    $\mathfrak{q} \in \C$, see \cite[proposition 1]{its2014}. This
      technique, which requires $\epsilon_1+\epsilon_2 = 0$, also generalizes
      \cite{Bershtein2017} to the $K$-theoretic Nekrasov partition function without
      additional matter fields.
  \end{rmk}

  \subsection{Conformal Blocks}

  \subsubsection{The Virasoro Algebra and Verma Modules}
  We consider the Virasoro algebra $Vir$. It is
  a complex Lie algebra with generators $L_n, n \in \Z$ and central element $C$
  satisfying 
  \begin{align*}
    [L_m,L_n] = (m-n) L_{m+n} + \frac{1}{12} (m^3-m) \delta_{m+n,0} C.
  \end{align*}
  Fix $c \in \C$, the so-called \emph{central charge}. It will remain unchanged
  for the remainder of this introduction. For $h \in \C$ the \emph{Verma module
  }$V_{h}$ for the Virasoro algebra of conformal dimension $h$ and central charge
  $c$ is a module generated by a vector $|h \ra$ satisfying $L_0 |h \ra = h |h
  \ra$, $C |h \ra = c |h \ra$ and $L_n |h \ra = 0$ for $n \geq 1$. A basis is
  given by $L_{\lambda} |h \ra := L_{-\lambda(1)} \cdots L_{-\lambda(l)} |h \ra$,
  where $\lambda = (\lambda(1), \dots, \lambda(l))$ is a partition. The Verma
  module comes with a bilinear form, the \emph{Shapovalov form}, $S : V_h \otimes
  V_h \to \C$ characterized by $S( L_n x, y ) = S(x, L_{-n} y)$, for $x,y \in
  V_h$ and $S(|h \ra, |h \ra) = 1$. 

  \subsubsection{Intertwiners}
  Fix complex numbers $h, h_1, h_2$ and set $a = h_2 - h_1 - h$. 
  Given two Verma modules $V_{h_1}$ and $V_{h_2}$ for the same central charge $c$,
  an \emph{intertwiner} $\phi(z) = \sum_{n \in \Z} \phi_n z^{n+a}$
  between $V_{h_1}$ and $V_{h_2}$ of conformal dimension $h$ is a formal power
  series whose coefficients $\phi_n$ are linear maps $V_{h_1} \to V_{h_2}$ such
  that 
  \begin{align}
    \label{eq:intertwiner}
    [L_n, \phi(z)] = ( z^{n+1} \partial_z + (n+1) h z^n ) \phi(z).
  \end{align}
  In conformal field theory, this property models transformations of fields under
  infinitesimal conformal transformations. Conformal blocks model vacuum
  expectation values of such fields. Mathematically they are defined as follows:

  \subsubsection{Conformal Blocks}
  Let $h_l, H_l, h_m, H_r, h_r$ be complex numbers. Let $\phi_l(w)$ be an
  intertwiner of conformal dimension $H_l$ from $V_{h_m}$ to $V_{h_l}$. Let
  $\phi_r(z)$ be an intertwiner of conformal dimension $H_r$ from $V_{h_r}$ to
  $V_{h_m}$. The object
  \begin{align*}
    \la h_l | \phi_l(1) \phi_r(z) | h_r \ra :=
    \lim_{w \to 1}
    S\bigg( |h_l \ra, \phi_l(w) \Big( \phi_r(z) |h_r \ra \Big) \bigg)
  \end{align*}
  is called a four-point conformal block. It defined as $z^{h-h_r-H_r}$ times a
  formal power series in $z$, whose complex coefficients are in
  principle determined by equation \eqref{eq:intertwiner} up to normalization. The
  limit is taken for each coefficient separately. 

  The AGT relation allows us to express the conformal block as a Nekrasov
  partition function: One first has to introduce the \emph{Liouville
    parametrization} of the conformal dimensions and the central charge. We pick
  $b \in \C$ with 
  \begin{align*}
    c=1+6(b+b^{-1})^2,
  \end{align*}
  to parametrize the central charge. The conformal dimensions $h_l, h_m, h_r$ of
  the Verma modules $V_{h_l}$, $V_{h_m}$ and $V_{h_r}$ are parametrized as
  \begin{align*}
    h_l = \frac{(b+b^{-1})^2}{4} - P_l^2, \qquad h_m = \frac{(b+b^{-1})^2}{4} - P_m^2, \qquad h_l =
    \frac{(b+b^{-1})^2}{4} - P_l^2,
  \end{align*}
  where $P_l,P_m,P_r \in \C$. The conformal dimensions $H_l$ and $H_r$ of the
  intertwiners $\phi_l(w)$ and $\phi_r(z)$ are parametrized as
  \begin{align*}
    H_{r} &= \alpha_{r} ( b+b^{-1} - \alpha_{r} ), &
    H_{l} &= \alpha_{l} ( b+b^{-1} - \alpha_{l} ),
  \end{align*}
  where $\alpha_r, \alpha_l \in \C$.
  We have \cite{albafateev2011}
  \begin{align*}
    \la h_l | \phi_l(1) \phi_r(z) | h_r \ra =
    (1-z)^{2 \alpha_r( b+b^{-1} - \alpha_l)}
    \sum_{n \geq 0} z^n F_n(b, \alpha_r, \alpha_l, P_r, P_m, P_l),
  \end{align*}
  with
  \begin{align*}
    &F_n(b, \alpha_r, \alpha_l, P_r, P_m, P_l) \\
    =& 
    \sum_{|Y_1| + |Y_2| = n} 
    \frac
    {
      Z_{bif}\Big( \alpha_r \Big| P_r, (\es, \es); P_m, (Y_1, Y_2) \Big) 
      \,
      Z_{bif}\Big( \alpha_l \Big| P_m, (Y_1, Y_2);  P_l, (\es,\es)
      \Big)
    }
    {
      Z_{bif}\Big( 0 \Big| P_m, (Y_1, Y_2); P_m, (Y_1, Y_2) \Big)
    }.
  \end{align*}
  Here the value of $Z_{bif}(\alpha | P', \vec{Y}';  P, \vec{Y})$ for
  pairs of partitions $\vec{Y}, \vec{Y}'$ and complex numbers $P,
  P'$ and $\alpha$ is given by
  \begin{align*}
     Z_{bif}(\alpha | P', \vec{Y}';  P, \vec{Y} )
    = 
    \prod_{i,j = 1}^2 &\Bigg(
    \prod_{s \in Y_i} 
    ( P_j' - P_i + b (l_{Y'_j}(s) +1 ) - b^{-1} a_{Y_i}(s) -\alpha)
    \\
    &\prod_{t \in Y'_j}
    ( P_j' - P_i - b l_{Y_i}(t) + b^{-1} (a_{Y'_j}(t) + 1) -\alpha) \Bigg),
  \end{align*}
  where, on the right hand side, $P_i$ is the $i$-th component of the vector
  $\vec{P} = (P,-P)$ and similarly for $P'_i$.

  \subsubsection{Attempt at an Estimate}

  A simple calculation gives
  \begin{align*}
      &Z_{bif}\Big( \alpha_r \Big| P_r, (\es, \es); P_m, (Y_1, Y_2) \Big) 
      \,
      Z_{bif}\Big( \alpha_l \Big| P_m, (Y_1, Y_2);  P_l, (\es,\es)
      \Big)
      \\
    =& 
    \prod_{i=1}^2 \prod_{(x,y) \in Y_i} \prod_{m=1}^4 ( 
    (P_m)_i + b (x-1) + b^{-1} (y-1) + v_m)
    ,
  \end{align*}
  where $v_1 = \alpha_r + P_r, v_2 = \alpha_r - P_r, v_3 = -\alpha_l + b +
  b^{-1} + P_l$ and
  $v_4 = - \alpha_l + b + b^{-1} - P_l.$ Under the identifications
  \begin{align*}
    a_\alpha = P_\alpha,\quad \alpha = 1,2, \qquad \ee_1 = b, \ee_2 = b^{-1}, 
    \qquad w_m = v_m, \quad m = 1, \dots, 4
  \end{align*}
  and using equation \eqref{eq:residueevaluation} in Theorem
  \ref{thm:residuecalculation},
  we recover $F_n(b, \alpha_r, \alpha_l, P_r, P_m, P_l)$
  from $Z_n(\ee_1, \ee_2, \vec{a}, \vec{w}; \lambda)$ in the limit $\lambda \to
  0$. 
  The conditions formulated in Theorem \ref{thm:nekrasovconvergence} on $b$ are
  $\Re b > 0$ and either $b > 0$ or $|b| = 1$. The ones on $P$ are
  $|\Re P| < \frac12 \Re b$ and $|\Re P | < \frac12  \Re b^{-1} $.
  Condition \eqref{eq:grid1} translates, for $\lambda = 0$, to
  \begin{align*}
    P \notin \frac{b}{2} \Z + \frac{b^{-1}}{2} \Z  \text{ and } 
    b^2 \notin \Q_{\geq 0}.
  \end{align*}
  The condition on $P$ excludes the minimal
  models. The conditions on $b$ restrict the central charge of the theory to the
  interval $(1,\infty)$.
  Since $\lambda^{2r-s} = 1$, Theorem \ref{thm:nekrasovconvergence} seems to suggest that the conformal block
  $\la h_l | \phi_l(1) \phi_r(z) | h_r \ra$ is analytic in $z$ for $|z| < 1$.
  However we neither know that the convergence of the Nekrasov partition
  function is uniform in $\lambda$, nor do we know that the convergence
  $Z_n(\ee_1, \ee_2, \vec{a}, \vec{w}; \lambda) \to
  F_n(b, \alpha_r, \alpha_l, P_r, P_m, P_l)$  as $\lambda \to 0$ is uniform in $n$.

  \appendix

  \section{Evaluation by Iterated Residues}
  \label{sec:residueproof}

  In this section we prove Theorem \ref{thm:residuecalculation} by evaluating the 
integral by residues.
    In a first step, we find the position of all poles whose residues contribute to the
    integral and show that they are simple. In a second step we evaluate the
    residues.

We show that $\ZZ_n(\vec{u};\vec{p})$ is a sum 
  \begin{align}
    \label{eq:residuesum}
    \ZZ_n(\vec{u};\vec{p}) &= \sum_{|\vec{Y}| = n }
    \ZZ_{\vec{Y}}(\vec{u};\vec{p}),
  \end{align}
  of iterated, simple residues
  \begin{align}
    \label{eq:residueexpression}
    \ZZ_{\vec{Y}}(\vec{u};\vec{p}) :=& \bigg( \frac{1-q_1q_2}{(1-q_1)(1-q_2)}\bigg)^n
    \prod_{j=1}^n \prod_{m=1}^s(\hat{z}_j-p_m)
    \\\nonumber
    &
    \lim_{\substack{z_j \to \hat{z}_j \\ j=1,\dots,n}}
    \;
    \Bigg(
    \prod_{j=1}^n 
    (z_j - \hat{z}_j)
    \quad
    \frac{ \mathcal{I}(z_1, \dots, z_n;\vec{u}) }
    {z_1 \cdots z_n}
    \Bigg),
  \end{align}
  where 
  \begin{align}
    \label{eq:unorderedresidues}
    \{\hat{z}_1,\dots,\hat{z}_n\} = 
    \{z^{\alpha}_{x,y} : (x,y) \in Y_\alpha, \alpha = 1, \dots, r \},
  \end{align}
  in any
  order.

    We suppose that $|q_1| = |q_2| + \delta < 1$, where $\delta > 0$ is small
    enough so that $|q_2| > |q_1|^2 > | q_1 |^3 > \cdots$. The general case
    follows from analytic continuation.

    We evaluate the integral
    $\ZZ_n(\vec{u};\vec{p})$ by taking iterated residues. We start with a
    slightly more general integral. Let $U,W$ be two finite sets of complex
    numbers and $f(z_1,\dots,z_n)$ a symmetric function, analytic on the closed
    ball in $\C^n$ with radius $\rho$. Our integral is of the form
    \begin{align}
      \nonumber
      \int_{C_\rho} \frac{dz_n}{2\pi i}
      \dots
      \int_{C_\rho} \frac{dz_1}{2\pi i}
      &
      I(z_1,\dots,z_n)
      \\
      =   
      \label{eq:upper}
      \int_{C_\rho} \frac{dz_n}{2\pi i}
      \frac
      {
        \prod_{w \in W}
        (z_n -w)
      }
      {
        \prod_{u \in U}
        (z_n - u)
      }
      &
      \prod_{n<k}
      \frac
      {
        (z_n - z_k)^2
        (z_n - q_1 q_2 z_k)
        (z_n - q_1^{-1} q_2^{-1} z_k)
      }
      {
        (z_n - q_1 z_k)
        (z_n - q_2 z_k)
        (z_n - q_1^{-1} z_k)
        (z_n - q_2^{-1} z_k)
      }
      \\\nonumber
      &
      \vdots
      \\\nonumber
      \int_{C_\rho} \frac{dz_2}{2\pi i}
      \frac
      {
        \prod_{w \in W}
        (z_2 -w)
      }
      {
        \prod_{u \in U}
        (z_2 - u)
      }
      &
      \prod_{2<k}
      \frac
      {
        (z_2 - z_k)^2
        (z_2 - q_1 q_2 z_k)
        (z_2 - q_1^{-1} q_2^{-1} z_k)
      }
      {
        (z_2 - q_1 z_k)
        (z_2 - q_2 z_k)
        (z_2 - q_1^{-1} z_k)
        (z_2 - q_2^{-1} z_k)
      }
      \\\nonumber
      \int_{C_\rho} \frac{dz_1}{2\pi i}
      \frac
      {
        \prod_{w \in W}
        (z_1 -w)
      }
      {
        \prod_{u \in U}
        (z_1 - u)
      }
      &
      \prod_{1<k}
      \frac
      {
        (z_1 - z_k)^2
        (z_1 - q_1 q_2 z_k)
        (z_1 - q_1^{-1} q_2^{-1} z_k)
      }
      {
        (z_1 - q_1 z_k)
        (z_1 - q_2 z_k)
        (z_1 - q_1^{-1} z_k)
        (z_1 - q_2^{-1} z_k)
      }\\ \nonumber &
      f(z_1,\dots,z_n).
    \end{align}
    where
    $I(z_1,\dots,z_n)$ is given by
    \begin{align*}
      f(z_1,\dots,z_n) 
      \prod_{j=1}^n
      \Bigg(&
      \frac
      {
        \prod_{w \in W}
        (z_j -w)
      }
      {
        \prod_{u \in U}
        (z_j - u)
      }\\ &\times
      \prod_{j<k}
      \frac
      {
        (z_j - z_k)^2
        (z_j - q_1 q_2 z_k)
        (z_j - q_1^{-1} q_2^{-1} z_k)
      }
      {
        (z_j - q_1 z_k)
        (z_j - q_2 z_k)
        (z_j - q_1^{-1} z_k)
        (z_j - q_2^{-1} z_k)
      }
      \Bigg).
    \end{align*}
    When integrating $z_1$ we can only pick up residues at $q_1 z_j$, $q_2 z_j$
    or $\hat{z}_1 = u,$ for some $u \in U$ with $ |u| \leq \rho$. 

    Assume that we pick up a residue at $\hat{z}_1 = q_i z_j$ for some $i \in
    \{1,2\}$ and $j \in \{2, \dots, n\}$. The other case will be treated later
    on.  Let $s \in \{1,2\}$ be the index complementary to $i$. The residue is
    simple.
    After renaming $z_j
    \leftrightarrow z_2$ we obtain
    \begin{align*}
      &
      \int_{C_\rho} \frac{dz_n}{2\pi i}
      \dots
      \int_{C_\rho} \frac{dz_2}{2\pi i}
      f(q_i z_2,z_2,\dots,z_n) 
      \\
      & \times
      \prod_{j=3}^n
      \Bigg(
      \frac
      {
        \prod_{w \in W}
        (z_j -w)
      }
      {
        \prod_{u \in U}
        (z_j - u)
      }
      \prod_{2<j<k}
      \frac
      {
        (z_j - z_k)^2
        (z_j - q_1 q_2 z_k)
        (z_j - q_1^{-1} q_2^{-1} z_k)
      }
      {
        (z_j - q_1 z_k)
        (z_j - q_2 z_k)
        (z_j - q_1^{-1} z_k)
        (z_j - q_2^{-1} z_k)
      }
      \Bigg)
      \\
      & \times
      \frac
      {
        \prod_{w \in W}
        (z_2 -w)(q_i z_2 - w)
      }
      {
        \prod_{u \in U}
        (z_2 - u)(q_i z_2 -u)
      }
      \\
      & \times\prod_{2<k}
      \frac
      {
        (z_2 - z_k)
        (z_2 - q_i q_s z_k)
        (z_2q_i - z_k)
        (z_2q_i - q_i^{-1} q_s^{-1} z_k)
      }
      {
        (z_2 - q_i z_k)
        (z_2 - q_s^{-1} z_k)
        (z_2q_i - q_s z_k)
        (z_2q_i - q_i^{-1} z_k)
      }
      \\
      &
      \times
      q_i z_2 
      \frac
      {
        (q_i - 1)^2 
        (1 - q_s)
        (q_i - q_i^{-1} q_s^{-1})
      }
      {
        (q_i - q_s)
        (q_i - q_i^{-1})
        (q_i - q_s^{-1})
      }.
    \end{align*}
    We also have used Fubini's theorem to permute the order of integration, swapping the
    integrals over $dz_2$ and $dz_j$. This is permitted since the integrand does
    not have poles on the integration contours. Here we use the small perturbation of $|q_1|$ and $|q_2|$.  We note that this result does not
    depend on $j$ anymore, which also follows from the symmetry of
    $I(z_1,\dots,z_n)$. So we can suppose without loss of generality $j=2$.

    We see that now we can pick up residues at $\hat{z}_2 = u_0, q_i^{-1} u_0,$
    for some $u_0$ such that $\hat{z}_2$ lies inside the integration contour or at
    $\hat{z}_2 = q_i z_j$, where the $i$ is the same as in the previous integral.
    It appears like there is a residue coming from the pole $(z_2 q_i - q_s
    z_j)$ if $i=1$. However, all residues coming from these poles cancel each
    other in the sum: One directly calculates
    \begin{align*}
      \lim_{z_k \to q_1^{-1} q_2 z_l} (z_k - q_1^{-1} q_2 z_l) \sum_{i=1,2}
      \sum_{j=2}^n \Res_{\hat{z}_1 = q_i z_j} I(z_1,\dots,z_n) = 0.
    \end{align*}


    Assume that we pick $\hat{z}_2 = q_i z_j$, again assume without loss of
    generality $j=3$. The residue is simple. The integral over $z_3$ becomes
    \begin{align*}
      \int_{C_\rho} \frac{dz_3}{2\pi i}
      &\prod_{l=0}^2
      \frac
      {
        \prod_{w \in W}
        (z_3q_i^l -w)
      }
      {
        \prod_{u \in U}
        (z_3q_i^l - u)
      }
      \\
      &
      \prod_{3<k}
      \frac
      {
        (z_3 - z_k)
        (z_3 - q_i q_s z_k)
        (z_3q_i^2 - z_k)
        (z_3q_i^2 - q_i^{-1} q_s^{-1} z_k)
      }
      {
        (z_3 - q_i z_k)
        (z_3 - q_s^{-1} z_k)
        (z_3q_i^2 - q_s z_k)
        (z_3q_i^2 - q_i^{-1} z_k)
      } \\
      &
      q_i^3 z_3^2
      \frac
      {
        (q_i - 1)^3
        (1 - q_s)^2
        (q_i^2 - q_s^{-1})
        (q_i^3 - q_s^{-1} )
      }
      {
        (q_i - q_s)
        (q_i - q_s^{-1})^2
        (q_i^2 - q_s )
        (q_i^3 - 1 )
      }
      f(q_i^2 z_3, q_i z_3, z_3, \dots,z_n).
    \end{align*}
    In the next step, we can again only pick up simple residues at $\hat{z}_3 =
    q_i z_4$ or $\hat{z}_3 = q_i^{-l} u_0$ with $l \in \{0,1,2\}$ and $u_0 \in U$
    such that $\hat{z}_3$ lies inside the integration contour.

    Assume that we have picked and evaluated simple residues at
    $(\hat{z}_1,\dots,\hat{z}_{J-1})$ during the first $J-1$ integrations with
    $\hat{z}_j = q_i^{J-j} z_J$, $j = 1, \dots, J-1$, where $i \in \{1,2\}$ is
    fixed. Let $s \in \{1,2\}$ be the other index. The integral over $z_J$ has the
    integrand
    \begin{align}
      \label{eq:recursiveintegrand}
      &\prod_{l=0}^{J-1}
      \frac
      {
        \prod_{w \in W}
        (z_J q_i^l -w)
      }
      {
        \prod_{u \in U}
        (z_J q_i^l - u)
      }\\ \nonumber & \times
      \prod_{J<k}
      \frac
      {
        (z_J - z_k)
        (z_J - q_i q_s z_k)
        (z_Jq_i^{J-1} - z_k)
        (z_Jq_i^{J-1} - q_i^{-1} q_s^{-1} z_k)
      }
      {
        (z_J - q_i z_k)
        (z_J - q_s^{-1} z_k)
        (z_Jq_i^{J-1} - q_s z_k)
        (z_Jq_i^{J-1} - q_i^{-1} z_k)
      }\\
      \nonumber
      & \times
      z_J^{J-1}
      q_i^{\frac12 (J+2)(J-1)} 
      \frac
      {
        (q_i  - 1)^{J}
        (1 - q_s )^{J-1}
      }
      {
        (q_i  - q_s^{-1} )^{J-1}
        (q_i^{J} - 1)
      }\\ \nonumber & \times
      \prod_{j=1}^{J-1}
      \frac
      {
        (q_i^{j} - q_i^{-1} q_s^{-1} )
      }
      {
        (q_i^{j} - q_s )
      }
      f(q_i^{J-1} z_J, \dots, z_J, \dots, z_n),
    \end{align}
    whereas the integrands for $z_{J+1}, \dots, z_n$ in equation \eqref{eq:upper}
    remain the same. Now pick a residue at $\hat{z}_J = q_i^{-l_0} u_0$ inside
    the integration contour with $u_0 \in U$ and $l_0 \in \{0, \dots, J-1\}$ and
    suppose the pole is simple. So we have followed the evaluation steps
    \begin{align}
      \label{eq:evalsteps}
      \hat{z}_1 = q_i z_2, \qquad \hat{z}_2 = q_i z_3, \qquad \dots, \qquad
      \hat{z}_{J-1} = q_i z_J, \qquad \hat{z}_J = q_i^{-l_0} u_0.
    \end{align}
    leading to the residue strip
    \begin{align}
      \label{eq:residuestrip}
      (\hat{z}_1,\dots,\hat{z}_J) = q_i^{-l_0} (q_i^{J-1} u_0, \dots, u_0).
    \end{align}
    When we evaluate the final residue we
    get the same integral expression we started with as in equation
    \eqref{eq:upper}, except that we only have the variables $z_{J+1}, \dots,
    z_n$, the sets $U,W$ are changed to
    \begin{align*}
      U' &= U \cup \{ q_i^{-1} \hat{z}_J, q_{s} \hat{z}_J, q_s^{-1}
      \hat{z}_1, q_i \hat{z}_1 \} &
      W' &=  W \cup \{ \hat{z}_J, q_i^{-1} q_{s}^{-1} \hat{z}_J,    
     \hat{z}_1, q_i q_s \hat{z}_1 \},
     \end{align*}
     and $f(z_1,\dots,z_n)$ gets replaced by
     \begin{align*}
       f'(z_{J+1},\dots,z_n)=
       f(q_i^{J-1} q_i^{-l_0} u_0, \dots, q_i^{-l_0} u_0, z_{J+1}, \dots, z_n).
     \end{align*}
     This function is again symmetric. 
     Moreover, we have accumulated a prefactor with the value 
     \begin{align}
       \label{eq:residueprefactor}
       &
       \prod_{l=0}^{J-1}
       \frac
       {
         \prod_{w \in W}
         (q_i^{l-l_0} u_0 -w)
       }
       {
         \prod_{u \in U\setminus \{u_0\}}
         (q_i^{l-l_0} u_0 - u)
       }
       \prod_{\substack{l=0 \\ l\neq l_0}}^{J-1}
       \frac{1}
       {
         (q_i^{l-l_0} - 1)
       }
       q_i^{\frac12 (J+2)(J-1) - l_0 J} 
       \\\nonumber
       &
       \frac
       {
         (q_i  - 1)^{J}
         (1 - q_s )^{J-1}
       }
       {
         (q_i  - q_s^{-1} )^{J-1}
         (q_i^{J} - 1)
       }
       \prod_{j=1}^{J-1}
       \frac
       {
         (q_i^{j} - q_i^{-1} q_s^{-1} )
       }
       {
         (q_i^{j} - q_s )
       },
    \end{align}
    coming from the residue evaluation. We see that the evaluation of our
    integral happens in stages, where one evaluates a strip of residues. We
    draw the strip
    \eqref{eq:residuestrip} for $i=2$ as follows:
    \begin{align}
      \label{eq:residuediagram}
      \begin{tikzpicture}[scale=0.5]
        \draw [fill=zero, draw=none] (-4,1) rectangle (-3,2);
        \draw [fill=pole, draw=none] (-4,0) rectangle (-3,1);
        \draw [fill=pole, draw=none] (-3,-1) rectangle (-2,0);
        \draw [fill=zero, draw=none] (3,-1) rectangle (4,0);
        \draw [fill=pole, draw=none] (2,1) rectangle (3,2);
        \draw [fill=pole, draw=none] (3,0) rectangle (4,1);
        \draw [very thick, fill=zero] (-3,0) rectangle (-2,1)
        node[pos=.5]{$\hat{z}_J$}; 
        \draw[very thick] (-2,0) rectangle (-1,1)
        node[pos=.5]{$\cdots$} ;
        \draw[very thick] (-1,0) rectangle (0,1) node[pos=.5]{$\cdots$} ;
        \draw[very thick] (0,0) rectangle (1,1) node[pos=.5]{$u_0$} ;
        \draw[very thick] (1,0) rectangle (2,1) node[pos=.5]{$\cdots$} ;
        \draw [very thick, fill=zero] (2,0) rectangle (3,1)
        node[pos=.5]{$\hat{z}_1$}; \end{tikzpicture}
    \end{align} 
    The positive $q_i$ direction goes from west to east.
    The residues go from east to west over the residue strip
    \eqref{eq:residuestrip}.  Here we have also indicated the poles (red) and
    zeros (blue) such a strip adds to the sets $U$ and $W$. For $i=1$ the strip
    is drawn vertically, with positive $q_i$ direction from north to south.

    All formulae remain valid for $J=1$, i.e. when we directly pick up a residue
    at $\hat{z}_1 = u_0$ for some $u_0 \in U$. We call such a strip of length one
    a \emph{box}. It is drawn as
    
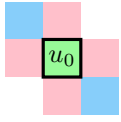
\begin{align}
      \label{eq:residuebox}
      \begin{tikzpicture}[scale=0.5]
        \draw [fill=zero, draw=none] (-1,1) rectangle (0,2);
        \draw [fill=zero, draw=none] (1,-1) rectangle (2,0);
        \draw [fill=pole, draw=none] (1,0) rectangle (2,1);
        \draw [fill=pole, draw=none] (-1,0) rectangle (0,1);
        \draw [fill=pole, draw=none] (0,1) rectangle (1,2);
        \draw [fill=pole, draw=none] (0,-1) rectangle (1,0);
        \draw [very thick, fill=doublezero] (0,0) rectangle (1,1) 
        node[pos=.5]{$u_0$};
      \end{tikzpicture}
    \end{align} 
    Here we have used the color green to mark a zero of order two.
    We want to compare the result
    of the evaluation process \eqref{eq:evalsteps} 
    of the above residues for general $l_0 \in \{0,\dots, J-1\}$
    to the result from the procedure where we evaluate the same final
    residues \eqref{eq:residuestrip} by repeated use of the case $J=1$. Now, the
    order of the residues will be different. It is clear that the sets $U'$ and
    $W'$ we end up with agree for both procedures. Moreover, the new symmetric
    function $f'$ also agrees, since the original function $f$ is symmetric.
    However the prefactor \eqref{eq:residueprefactor} only agrees up to a sign, as
    we will now see.

    We first treat the case, where we go in positive $q_i$ direction, starting at
    a pole at some $u_1$: Suppose we pick up residues using the evaluation process 
    \begin{align}
      \label{eq:altres}
      \hat{z}_1 = u_1, \qquad \hat{z}_2 = q_i u_1, \qquad \dots \qquad \hat{z}_K = q_i^{K-1}
      u_1.
    \end{align}
    For $k = 0, \dots, K$, let $U^+_k$ and $W^+_k$ denote the sets $U$ and $W$ after evaluating the
    residues at $\hat{z}_1,\dots,\hat{z}_{k}$.
    In the first step we pick up the residue at $\hat{z}_1 = u_1$, the sets $U =
    U^+_0$ and $W = W^+_0$
    get changed to
    \begin{align*}
      U^+_1 &= \Big( U \setminus \{u_1\} \Big) \cup \{ q_i^{-1} u_1, q_{s} u_1, q_s^{-1} u_1, q_i u_1 \}
\\
      W^+_1 &= W \cup \{ u_1, q_i^{-1} q_{s}^{-1} u_1, q_i q_s u_1 \},
    \end{align*}
    the symmetric function is $f(u_1,z_2,\dots,z_n)$ 
    and we get the prefactor
    \begin{align*}
      \frac
      {
        \prod_{w \in W^+_0}
        (u_1 -w)
      }
      {
        \prod_{u \in U^+_0 \setminus \{u_1\}}
        (u_1 - u)
      }.
    \end{align*}
    By induction, we see that after evaluating all residues \eqref{eq:altres}, we
    have changed the original sets $U,W$ to
    \begin{align*}
      U^+_{K} & = \Big( U \setminus \{u_1\} \Big) \cup \{ q_i^{-1} u_1, q_{s}
      u_1, q_s^{-1} q_i^{K-1} u_1, q_i q_i^{K-1} u_1 \} \\
      W^+_{K} & =  W \cup \{q_i^{-1} q_{s}^{-1} u_1, q_i^{K-1} u_1, q_i q_s
      q_i^{K-1} u_1 \},
    \end{align*}
    the function $f(z_1,\dots,z_n)$ gets replaced by
    \begin{align*}
      f_K^+(z_{K+1}, \dots, z_n) = f(u_1, \dots, q_i^{K-1} u_1, z_{K+1},\dots,
      z_n) \end{align*}
    and we have accumulated the prefactor
    \begin{align}
      \label{eq:positiveaccumulation}
      \prod_{k=0}^{K-1}
      &\frac
      {
        \prod_{w \in W^+_k}
        (q_i^k u_1 -w)
      }
      {
        \prod_{u \in U^+_k \setminus \{ q_i^{k} u_1\}}
        (q_i^k u_1 - u)
      }
       =
      \prod_{k=0}^{K-1}
      \frac
      {
        \prod_{w \in W}
        (q_i^k u_1 -w)
      }
      {
        \prod_{u \in U \setminus \{ u_1\}}
        (q_i^k u_1 - u)
      }
      \prod_{k=1}^{K-1}
      \frac
      {
        1
      }
      {
        (q_i^{k}  - 1  )
      }
      \\\nonumber
      &
      q_i^{\frac12 (K+2)(K-1)}
      \frac
      {
        (q_i  - 1  )^{K}
        (1 - q_s  )^{K-1}
      }
      {
        (q_i^K -1)
        (q_i  - q_s^{-1} )^{K-1}
      }
      \prod_{k=1}^{K-1}
      \frac
      {
        (q_i^k  - q_i^{-1} q_{s}^{-1}  )
      }
      {
        (q_i^k  - q_{s}  )
      }.
    \end{align}

    Next, we treat the case where we go in negative $q_i$ direction starting at
    some pole $u_2$. Suppose we pick up residues using the evaluation steps
    \begin{align}
      \label{eq:altres2}
      \hat{z}_1 = u_2, \qquad \hat{z}_2 = q_i^{-1} u_2, \qquad \dots \qquad
      \hat{z}_M = q_i^{-M+1}
      u_2.
    \end{align}
    For $m = 0, \dots, M$, let $U^-_m$ and $W^-_m$ denote the sets $U$ and $W$ after evaluating the residues
    at $\hat{z}_1,\dots,\hat{z}_{m}$.
    In the first step we pick up the residue at $\hat{z}_1 = u_2$, the sets $U =
    U^-_0$ and $W = W^-_0$
    get changed to
    \begin{align*}
      U^-_1 &= \Big( U \setminus \{u_2\} \Big) \cup \{ q_i^{-1} u_2, q_{s} u_2, q_s^{-1} u_2, q_i u_2 \} \\
      W^-_1 &= W \cup \{ u_2, q_i^{-1} q_{s}^{-1} u_2, q_i q_s u_2 \},
    \end{align*}
    the symmetric function is $f(u_2,z_2,\dots,z_n)$ 
    and we get the prefactor
    \begin{align*}
      \frac
      {
        \prod_{w \in W^-_0}
        (u_2 -w)
      }
      {
        \prod_{u \in U^-_0 \setminus \{u_2\}}
        (u_2 - u)
      }.
    \end{align*}
    Again after evaluating all residues \eqref{eq:altres2}, we have changed the
    original sets $U,W$ to
    \begin{align*}
      U^-_M & = \Big( U \setminus \{u_2\} \Big) \cup \{ q_i u_2, q_{s}^{-1} u_2,
      q_s q_i^{-M+1} u_2, q_i^{-1} q_i^{-M+1} u_2 \} \\
      W^-_M & =  W \cup \{q_i q_{s} u_2, q_i^{-M+1} u_2, q_i^{-1} q_s^{-1} q_i^{-M+1}
      u_2 \},
    \end{align*}
    the function $f(z_1,\dots,z_n)$ is replaced by
    \begin{align*}
      f^-_M(z_{M+1},\dots,z_n) = f(u_2, \dots, q_i^{-M+1} u_2, z_{M+1},\dots, z_n)
    \end{align*}
    and get the prefactor
    \begin{align}
      \label{eq:minusprefactor}
      \prod_{m=0}^{M-1}
      &\frac
      {
        \prod_{w \in W^-_m}
        (q_i^{-m} u_2 -w)
      }
      {
        \prod_{u \in U^-_m \setminus \{ q_i^{-m} u_2\}}
        (q_i^{-m} u_2 - u)
      }
       =
      \prod_{m=0}^{M-1}
      \frac
      {
        \prod_{w \in W}
        (q_i^{-m} u_2 -w)
      }
      {
        \prod_{u \in U \setminus \{ u_2\}}
        (q_i^{-m} u_2 - u)
      }
      \prod_{m=1}^{M-1}
      \frac
      {
        1
      }
      {
        (q_i^{-m} - 1 )
      }
      \\\nonumber
      &
      q_i^{-\frac12 (M-1)M + M-1}
      \frac
      {
        (1 -q_i)^{M}
        (1 - q_s)^{M-1}
      }
      {
        (q_i - q_s^{-1} )^{M-1}
        (1 - q_i^M )
      }
      \prod_{m=1}^{M-1}
      \frac
      {
        (q_i^m - q_i^{-1} q_s^{-1} )
      }
      {
        (q_i^m - q_s )
      }.
    \end{align}

    Recall that were looking at the case where we pick the residues
    \eqref{eq:residuestrip} as a strip depicted in figure
    \eqref{eq:residuediagram} using the evaluation steps \eqref{eq:evalsteps} and
    we want to compare it to the procedure where we apply the case $J=1$
    repeatedly. Set $M = l_0+1$ and $K = J - l_0 - 1$.
    Picking residues according to the stepwise procedures with $u_2 = u_0$ and $u_1
    = q_i u_0$ we get the prefactor
    \begin{align}
      \label{eq:basevalue}
      &
      \prod_{k=0}^{K-1}
      \frac
      {
        \prod_{w \in W^-_M}
        (q_i^k u_1 -w)
      }
      {
        \prod_{u \in U^-_M \setminus \{ u_1\}}
        (q_i^k u_1 - u)
      }
      \prod_{m=0}^{M-1}
      \frac
      {
        \prod_{w \in W}
        (q_i^{-m} u_2 -w)
      }
      {
        \prod_{u \in U \setminus \{ u_2\}}
        (q_i^{-m} u_2 - u)
      }
      \\\nonumber
      &
      q_i^{\frac12 (K+2)(K-1)}
      q_i^{-\frac12 (M-1)M + M-1}
      \frac
      {
        (q_i  - 1  )^{K}
        (1 - q_s  )^{K-1}
      }
      {
        (q_i^K -1)
        (q_i  - q_s^{-1} )^{K-1}
      }
      \frac
      {
        (1 -q_i)^{M}
        (1 - q_s)^{M-1}
      }
      {
        (q_i - q_s^{-1} )^{M-1}
        (1 - q_i^M )
      }
      \\\nonumber
      &
      \prod_{k=1}^{K-1}
      \frac
      {
        (q_i^k  - q_i^{-1} q_{s}^{-1}  )
      }
      {
        (q_i^k  - q_{s}  )
      }
      \prod_{m=1}^{M-1}
      \frac
      {
        (q_i^m - q_i^{-1} q_s^{-1} )
      }
      {
        (q_i^m - q_s )
      }
      \prod_{k=1}^{K-1}
      \frac
      {
        1
      }
      {
        (q_i^{k}  - 1  )
      }
      \prod_{m=1}^{M-1}
      \frac
      {
        1
      }
      {
        (q_i^{-m} - 1 )
      }
      .
    \end{align}
    We call this factor the \emph{base value} for our strip of residues.  This factor
    equals the prefactor in equation \eqref{eq:residueprefactor} up to a sign of 
    \begin{align*}
      (-1)^{M+1}.
    \end{align*}
    (When we first evaluate the residues in positive $q_i$ direction and then in
    negative $q_i$ direction we get the same result.)

    We point out two special cases of our observation:
    \begin{enumerate}
      \item In the case $l_0 = 0$, where we place the whole strip of residues at
        $u_0$ and eastwards of $u_0$, it does not matter whether we pick the
        residues one by one or as a strip.
      \item In the case $l_0 = J-1$, where we place the whole strip of residues
        at $u_0$ and westwards of $u_0$, we get the base value up to a
        sign equal to $(-1)^{L+1}$ where $L$ is the length of the strip.
    \end{enumerate}
    \begin{claim}
      \label{claim:l0claim}
      Only strips with $l_0 = 0$ contribute to the integral.
    \end{claim}

    Suppose, we add a strip of residues of length $J$ with $l_0 > 0$, i.e. part
    of the residue diagram grows in negative $q_i$ direction. We claim that the
    total contribution of all possible processes to choose residues leading to
    the same residue strip is zero. We depict the strip as
    \begin{align}
      \label{eq:bulletstars}
      \begin{ytableau}
        \star &\star & \star & \star & \star & \dagger & \dagger & \dagger \\
      \end{ytableau}
      \\\nonumber
      \uparrow \hspace{5.1em}
      \\\nonumber
      u_0 \hspace{4.7em}
    \end{align}
    for $i=2$ and vertically for $i=1$. The number of stars equals $M = l_0 +1$
    and the number of daggers equals $K = J-M$. The residue at $u_0$ is located
    at the easternmost star. We draw the corresponding base value, corresponding
    to a stepwise procedure when picking up residues,
    as \begin{align*}
      \begin{ytableau}
        \\
      \end{ytableau}
      \;
      \begin{ytableau}
        \\
      \end{ytableau}
      \;
      \begin{ytableau}
        \\
      \end{ytableau}
      \;
      \begin{ytableau}
        \\
      \end{ytableau}
      \;
      \begin{ytableau}
        \\
      \end{ytableau}
      \;
      \begin{ytableau}
        \\
      \end{ytableau}
      \;
      \begin{ytableau}
        \\
      \end{ytableau}
      \;
      \begin{ytableau}
        \\
      \end{ytableau}
      \\
      \uparrow \hspace{6.1em}
      \\
      u_0 \hspace{5.7em}
    \end{align*}
    Of course, this picture does not fully specify the procedure: We can
    alternate between picking up residues to the east and to the west. We will
    deal with this multiplicity later.

    Another way to evaluate the residues is for example by first evaluating some
    of the residues as a strip with $l_0=J-2$ and then two as boxes and one final
    strip of length two to the east.
    We depict this as
    \begin{align}
      \label{eq:examplepartition}
      \begin{ytableau}
        \\
      \end{ytableau}
      \;
      \begin{ytableau}
        \\
      \end{ytableau}
      \;
      \begin{ytableau}
        \; & \;  & \; & \; \\
      \end{ytableau}
      \;
      \begin{ytableau}
        \; & \; \\
      \end{ytableau}
      \\\nonumber
      \uparrow \hspace{5.6em}
      \\\nonumber
      u_0 \hspace{5.em}
    \end{align}
    All possible ways to end up with the residues \eqref{eq:residuestrip} 
    define a partition of the strip \eqref{eq:bulletstars}
    into substrips. All those
    partitions yield the same final sets $U',W'$ and the same final function
    $f'$. The contributions of the residue evaluations differ by signs depending
    on the partition into substrips. To compute
    these signs we have to cut the partition eastwards of $u_0$, for example we
    substitute the partition \eqref{eq:examplepartition} by
    \begin{align*}
      \begin{ytableau}
        \\
      \end{ytableau}
      \;
      \begin{ytableau}
        \\
      \end{ytableau}
      \;
      \begin{ytableau}
        \; & \; & \; \\
      \end{ytableau}
      \;
      \begin{ytableau}
        \\
      \end{ytableau}
      \;
      \begin{ytableau}
        \; & \; \\
      \end{ytableau}
      \\
      \uparrow \hspace{5.9em}
      \\
      u_0 \hspace{5.5em}
    \end{align*}
    According to our computation above, we then have to count the number of strips
    of even length westwards of $u_0$. For each of those, we get a factor of
    $(-1)$. Then we have to sum over all partitions of the original string,
    weighted by the number of ways in which one can pick residues to obtain the
    partition. As we have already noticed, there are multiple ways since one can for example alternate between
    placing substrips at the west and east end. The total contribution of all
    procedures leading to the final residues \eqref{eq:residuestrip} is given by
    the base value times an expression
    \begin{align*}
      \sum_{\xi \in \Gamma} w(\xi) s(\xi),
    \end{align*}
    which we want to formalize. Here the set $\Gamma$ represents all the ways to
    cut the strip \eqref{eq:bulletstars} into substrips. The integer $w(\xi)$ is
    the
    weight corresponding to all possible procedures to pick residues leading to
    the choice of substrips $\xi$. 
    Finally the sign $s(\xi)$ is given by the number of substrips
    of even length to the west of $u_0$ after cutting directly to the east of
    $u_0$. We set 
    \begin{align*}
      \Gamma_J = \{ (J_1, \dots, J_N) : N \in \N, J_1, \dots, J_N \in \N : J_0 +
      \cdots + J_N = J \}
    \end{align*}
    to be the set of ordered partitions of the integer $J$. Cutting the strip
    \eqref{eq:bulletstars} into substrips according to $\vec{J} \in \Gamma_J$
    means cutting it into substrips of length $J_0, \dots, J_n$ from west to east.
    For example the cut strip \eqref{eq:examplepartition} corresponds to the
    element $\vec{J} = (1,1,4,2) \in \Gamma_8$. Let
    \begin{align*}
      \text{cut} : \Gamma_J \to \Gamma_{M}
    \end{align*}
    be the map realized by cutting the collection of substrips $\vec{J}$ directly
    to the east of $u_0$. Now we can formalize the sign $s(\vec{J})$ as 
    \begin{align*}
      s(\vec{J}) = \prod_a (-1)^{1+ \big( \text{cut}(\vec{J}) \big)_a}.
    \end{align*}
    In order to describe the weight $w(\vec{J})$, we have to take into account two
    effects: Firstly, for each substrip of length $J_a > 1$, we have to count the
    number of choices when selecting a residue $\hat{z}_a = q_i z_b$ that
    correspond to the
    freedom to choose any of the remaining variables $z_b$ for $b \in \{1, \dots,
    J\}$. Secondly, we can
    alternate between placing substrips to the west of the already chosen residues
    or to the east of them. We first have to account for the second effect. To model the second effect, we reorder the vectors
    $\vec{J}$. We define the map
    \begin{align*}
      \text{reo} : \Gamma_J \to \Gamma_J
    \end{align*}
    by demanding that, for 
    \begin{align*}
      \vec{A} = (A_1, A_2, \dots, A_N) =  
      \text{reo}(\vec{J}) = (B_0, B_1, \dots, B_b, C_1, \dots, C_c),
    \end{align*}
    the number $B_0$ 
    corresponds to the substrip containing $u_0$, the numbers $B_1$ up to $B_b$
    correspond to the substrips to the west, with increasing index in
    westward direction, and the numbers $C_1$ up to $C_c$ correspond to
    substrips to the east, with increasing index in eastward direction. The
    indices $b = b(\vec{J})$ and $c=c(\vec{J})$ depend on the collection of
    substrips $\vec{J} \in \Gamma_J$. For example, the cut
    strip \eqref{eq:examplepartition} yields $\text{reo}(\vec{J}) = (4,1,1,2)$,
    $b=2$ and $c=1$.
    We interpret the order of the components of $\text{reo}(\vec{J})$ as the
    order in which we place the residues. All other orderings are obtained by
    permuting the elements $(\vec{B}, \vec{C}) = (B_1, \dots, B_b, C_1, \dots,
    C_c)$ of this tuple, such that the order of the elements in $\vec{B}$ and
    in $\vec{C}$ remain the same. This means we can alternate between placing
    residues to the west and to the east. Let
    \begin{align*}
      S_{b,c} = \{ \sigma \in S_{b+c}: \sigma(1) < \cdots < \sigma(b), \sigma(b+1)
      < \cdots < \sigma(b+c) \}
    \end{align*}
    denote the set of $(b,c)$ shuffles. We define 
    \begin{align*}
      \sigma \vec{A} = \sigma 
      (B_0, B_1, \dots, B_b, C_1, \dots, C_c) = (B_0, \sigma( B_1,  \dots, B_b,
      C_1, \dots, C_c ) ),
    \end{align*}
    where $\sigma$ acts on $(B_1,\dots, B_b, C_1, \dots, C_c)$ in the usual sense. We see that all
    possible ways to place the substrips $\text{reo}(\vec{J})$ corresponds to
    the orbit of this tuple under $S_{r,s}$. Hence we have formalized the second
    effect. Now we can also formalize the first effect described above: Suppose
    we are placing the substrip $\sigma(\vec{A})_k$. Then the remaining substrips
    correspond to $\sigma(\vec{A})_{k+1} + \cdots + \sigma(\vec{A})_N$ variables.
    Each time we choose a pole at $\hat{z}_a = q_i z_b$ in $\sigma(\vec{A})_k$, we
    can choose $z_b$ to be one of those variables, or one of the variables in
    $\sigma(\vec{A})_k$, we have not chosen so far. Hence, we get
    \begin{align*}
      \prod_{h=1}^{\sigma(\vec{A})_j -1} 
      \Big( \sum_{k=j+1}^N \sigma(\vec{A})_k + h \Big)
    \end{align*}
    possibilities in total when placing $\sigma(\vec{A})_k$. In the case $\sigma(\vec{A})_k
    = 1$, we get an empty product which we interpret as one. This agrees with the
    observation that $\sigma(\vec{A})_k=1$ implies we pick a residue from the set
    $U$, which leaves no further choices in terms of the 
    remaining variables. Now we can combine the two effects to obtain the expression 
    \begin{align*}
      w(\vec{J}) = \sum_{\sigma \in S_{b,c}} \prod_{j=1}^N \bigg(
      \prod_{h=1}^{\sigma(\vec{A})_j -1} 
      \Big( \sum_{k=j+1}^N \sigma(\vec{A})_k + h \Big) \bigg).
    \end{align*}
    where $b=b(\vec{J})$, $c=c(\vec{J})$ and $\vec{A} = \text{reo}(\vec{J})$. 
    In order to prove claim \ref{claim:l0claim}, we show
    \begin{align}
      \label{eq:zerosum}
      &\sum_{\vec{J} \in \Gamma_J} s(\vec{J}) w(\vec{J})
      = 
      \sum_{\vec{J} \in \Gamma_J} 
      \Bigg( \prod_a (-1)^{1+ \big( \text{cut}(\vec{J}) \big)_a} \Bigg)
      \\
      \nonumber
      & \times
      \Bigg( \sum_{\sigma \in S_{b(\vec{J}),c(\vec{J})}} \prod_{j=1}^N \bigg(
      \prod_{h=1}^{\sigma(\text{reo}(\vec{J}))_j -1} \Big( \sum_{k=j+1}^N
      \sigma(\text{reo}(\vec{J}))_k + h \Big)
      \bigg) \Bigg) = 0
    \end{align}
    by exhibiting parts of this sum that cancel each other. Fix one $\vec{J} \in
    \Gamma_J$ with $B_b > 1$ for  
    \begin{align*}
      \vec{A} = \text{reo}(\vec{J}) = (B_0, B_1, \dots, B_{b-1}, B_b, C_1,
      \dots, C_c).  
    \end{align*}
    We compare this choice of substrips to 
    \begin{align*}
      \vec{A}' := & (B_0, B_1, \dots, B_{b-1}, B_b-1, 1, C_1, \dots, C_c) \\
      =: & (B'_0, B'_1, \dots, B'_{b+1}, C'_1, \dots, C'_c).
    \end{align*}
    Graphically, $\vec{A}'$ is obtained from $\vec{A}$ by cutting the strip at the
    westernmost box. Let $\vec{J}' \in \Gamma_J$ with $\vec{A}' =
    \text{reo}( \vec{J}')$. We claim that 
    \begin{align*}
      w(\vec{J}) s(\vec{J}) = - w(\vec{J}') s(\vec{J}').
    \end{align*}
    This suffices: There is a one to one correspondence between collections of
    substrips whose westernmost substrip has length one and those whose
    westernmost substrip has length greater than one. A bijection is given by
    $\vec{A} \mapsto \vec{A}'$. 
    It is clear from the definition of $s(\vec{J})$ that $\vec{J}$ and $\vec{J}'$
    have opposite sign, provided $l_0 > 0$. Hence it suffices to show that their weight is equal. We thus
    want to prove the equality in
    \begin{align*}
      \sum_{\sigma \in S_{b,c}}
      \prod_{j=1}^N &\bigg(
      \prod_{h=1}^{\sigma(\vec{A})_j -1} \Big( \sum_{k=j+1}^N
      \sigma(\vec{A})_k + h \Big)
      \bigg) \\ &=
      \sum_{\sigma \in S_{b+1,c}} \prod_{j=1}^{N+1} \bigg(
      \prod_{h=1}^{\sigma(\vec{A}')_j -1} \Big( \sum_{k=j+1}^{N+1}
      \sigma(\vec{A}')_k + h \Big)
      \bigg).
    \end{align*}
    We decompose
    \begin{align*}
      S_{b+1,c} = \bigcup_{\sigma \in S_{b,c}} \{ \hat{\sigma}_\lambda : \lambda =
      1, \dots, N+1-\sigma(r) \},
    \end{align*}
    where, for $j = 1, \dots, N+1$, 
    \begin{align*}
      \hat{\sigma}_\lambda(j) = 
      \begin{cases}
        \sigma(j), & j = 1, \dots, b \\
        \sigma(b) + \lambda, & j = b+1 \\
        \sigma(j-1), & j = b+2, \dots, \sigma(b) + \lambda\\
        \sigma(j-1) + 1, & j > \sigma(b)+ \lambda + 1
      .
      \end{cases}
    \end{align*}
    The parameter $\lambda$ tells us how much we delay the placement of the
    residue box we cut away from the westernmost residue strip. Fix $\sigma \in
    S_{b,c}$. We will show
    \begin{align}
      \label{eq:placement}
      \prod_{j=1}^N \bigg(&
      \prod_{h=1}^{\sigma(\vec{A})_j -1} \Big( \sum_{k=j+1}^N
      \sigma(\vec{A})_k + h \Big)
      \bigg) \\ \nonumber &=
      \sum_{\lambda = 1}^{N+1-\sigma(b)} 
      \prod_{j=1}^{N+1} 
      \bigg( \prod_{h=1}^{\hat{\sigma}_\lambda(\vec{A}')_j -1} \Big( \sum_{k=j+1}^{N+1}
      \hat{\sigma}_{\lambda}(\vec{A}')_k + h \Big)
      \bigg).
    \end{align}
    We set $B := \sigma(b)$. This is step in which we place $B'_b$.
    The tuples $\hat{\sigma}_{\lambda}(\vec{A}')$ and $\sigma(\vec{A})$ are
    related as follows:
    \begin{align*}
      {}&\Big({}
      {}\hat{\sigma}_{\lambda}(\vec{A}')_1, {}
      {}\dots, {}
      {}\hat{\sigma}_{\lambda}(\vec{A}')_{B-1}, {}
      {}\hat{\sigma}_{\lambda}(\vec{A}')_B,{}
      {}\hat{\sigma}_{\lambda}(\vec{A}')_{B+1},{}
      {}\dots,{}
      {}\hat{\sigma}_{\lambda}(\vec{A}')_{N+1} {}
      {}\Big)  \\
      =& 
      {}\Big({}
      {}\sigma(\vec{A})_1, {}
      {}\dots,{}
      {}\sigma(\vec{A})_{B-1},{}
      {}\sigma(\vec{A})_{B}-1,{}
      \big(
      {}\sigma(\vec{A})_{B+1}, {}
      {}\dots, 
      1,
      \dots, {}
      {}\sigma(\vec{A})_N{}
      \big)
      {}\Big){},
    \end{align*}
    where $1$ sits at index $\lambda$ inside the inner tuple.
    Using this description, we rewrite the right hand side of equation
    \eqref{eq:placement} as
    \begin{align*}
      &\sum_{\lambda=1}^{N+1-B}
      \prod_{j=1}^{B-1} 
      \bigg( 
      \prod_{h=1}^{\sigma(\vec{A})_j -1} 
      \Big( \sum_{k=j+1}^{N} \sigma(\vec{A})_k + h \Big)
      \bigg)
      \bigg( 
      \prod_{h=2}^{\sigma(\vec{A})_B -1} \Big( \sum_{k=B+1}^{N}
      \sigma(\vec{A})_k + h \Big)
      \bigg)
      \\
      &
      \prod_{j=B+1}^{B+\lambda-1} 
      \bigg( 
      \prod_{h=2}^{\sigma(\vec{A})_j } 
      \Big( \sum_{k=j+1}^{N} \sigma(\vec{A})_k + h \Big)
      \bigg)
      \prod_{j=B+\lambda}^{N} 
      \bigg( 
      \prod_{h=1}^{\sigma(\vec{A})_j -1} \Big( \sum_{k=j+1}^{N}
      \sigma(\vec{A})_k + h \Big)
      \bigg). 
    \end{align*}
    This equals the left hand side of equation \eqref{eq:placement} up to a factor
    of 
    \begin{align*}
      \Big( \sum_{k=B+1}^{N} \sigma(\vec{A})_k + 1 \Big)^{-1}
      \sum_{\lambda=1}^{N+1-B}
      \prod_{j=B+1}^{B+\lambda-1} 
      \Bigg(
      &\Big( \sum_{k=j+1}^{N} \sigma(\vec{A})_k + 1 \Big)^{-1}\\
      &\Big( \sum_{k=j+1}^{N} \sigma(\vec{A})_k + \sigma(\vec{A})_j \Big)
      \Bigg).
    \end{align*}
    This factor equals one, since for arbitrary tuples $(x_1, \dots, x_L)$, we
    have 
    \begin{align*}
      &\sum_{J=1}^L
      \prod_{j=1}^J \frac{\sum_{k=j}^L x_k}{\sum_{k=j+1}^L x_k + 1}  \\
      & = \frac{ \sum_{k=1}^L x_k }{ \sum_{k=2}^L x_k + 1 }
      \Bigg( 1 + \frac{ \sum_{k=2}^L x_k }{ \sum_{k=3}^L x_k + 1 }
      \\ &\times\bigg( \cdots  \Big(
      1 + \frac{ x_{k-2} + x_{k-1} + x_k}{ x_{k-1} + x_k + 1}
      \big( 1 + \frac{ x_{k-1} + x_k}{ x_k + 1} ( 1 + \frac{x_k}{1})
      \big)\Big) \cdots \bigg) \Bigg) \\
      & =\sum_{k=1}^L x_k.
    \end{align*}

    Hence we have established equation \eqref{eq:zerosum} and thus proved claim
    \ref{claim:l0claim}.  We note two consequences of our argument up to now:
    \begin{claim}
      \label{claim:boxes}
      Since we can only consider strips with $l_0=0$, we can only consider
      the cases were we take single boxes as residues. In particular, the order
      of their evaluation does not matter.
    \end{claim}
    \begin{claim}
      \label{claim:se}
      We can therefore discard poles at $q_i^{-1} u_0$ with $u_0 \in U$
      since they are either out of the integration contours or realizable by a
      string with $J=2$ and $l_0 = 1$ and hence part of a zero sum described
      above. 
    \end{claim}

    Now we see inductively that the residues $\hat{z}_j$ contributing to the
    integral are of the form $u_\alpha q_1^{x-1} q_2^{y-1}$ where $(x,y) \in
    Y_\alpha$ for some Young diagrams $Y_\alpha$. Because of the condition
    \begin{align*}
      u_\alpha u_\beta^{-1} \notin \{ q_1^x q_2^y : x,y \in \Z \}, \qquad \alpha
      \neq \beta
    \end{align*}
    the poles do not interact and 
    we may suppose $r=1$. The sets $U$ and $W$ we start with are then $U=\{u_1\}$
    and $W = \es$. We evaluate a residue at the pole $u_1$ changing the zeros and
    poles of the remaining integrand. We depict this process as
    \begin{align*}
      \text{the residue }
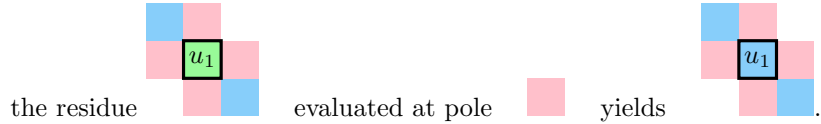
      \begin{tikzpicture}[scale=0.5]
        \draw [fill=zero, draw=none] (-1,1) rectangle (0,2);
        \draw [fill=zero, draw=none] (1,-1) rectangle (2,0);
        \draw [fill=pole, draw=none] (1,0) rectangle (2,1);
        \draw [fill=pole, draw=none] (-1,0) rectangle (0,1);
        \draw [fill=pole, draw=none] (0,1) rectangle (1,2);
        \draw [fill=pole, draw=none] (0,-1) rectangle (1,0);
        \draw [very thick, fill=doublezero] (0,0) rectangle (1,1) 
        node[pos=.5]{$u_1$};
      \end{tikzpicture}
      \quad
      \text{ evaluated at pole }
      \quad
      \begin{tikzpicture}[scale=0.5]
        \draw [fill=pole, draw=none] (0,0) rectangle (1,1);
      \end{tikzpicture}
      \quad
      \text{ yields }
      \quad
      \begin{tikzpicture}[scale=0.5]
        \draw [fill=zero, draw=none] (-1,1) rectangle (0,2);
        \draw [fill=zero, draw=none] (1,-1) rectangle (2,0);
        \draw [fill=pole, draw=none] (1,0) rectangle (2,1);
        \draw [fill=pole, draw=none] (-1,0) rectangle (0,1);
        \draw [fill=pole, draw=none] (0,1) rectangle (1,2);
        \draw [fill=pole, draw=none] (0,-1) rectangle (1,0);
        \draw [very thick, fill=zero] (0,0) rectangle (1,1) node[pos=.5]{$u_1$};
      \end{tikzpicture}.
    \end{align*}

    For the induction step,
    assume we have evaluated residues such that the evaluated residues fill a
    Young diagram. Also assume that the poles and zeros encoded in the sets $U,W$
    are located at the following places
    in the diagram: at each corner outside the diagram which is open to the south east there is a pole.
    To the northwest of each of those poles is another pole. At all
    south-easternmost boxes in the diagram there is a zero and another zero directly south east of
    it. Finally there
    is a zero at the coordinate $(-1,-1)$.  An example would be for instance
    \begin{align}
      \label{eq:diagramstep}
      \begin{tikzpicture}[scale=0.5]
        \fill[pole] (1,-4) rectangle (2,-3);
        \fill[pole] (3,-2) rectangle (4,-1);
        \fill[zero] (-1,1) rectangle (0,2);
        \fill[pole] (7,-1) rectangle (8,0);
        \fill[zero] (7,-2) rectangle (8,-1);
        \fill[zero] (3,-4) rectangle (4,-3);
        \draw[fill=pole, draw=none] (-1,-6) rectangle (0,-5)
        node[pos=.5]{$\times$}; 
        \fill[pole] (0,-7) rectangle (1,-6);
        \fill[zero] (1,-7) rectangle (2,-6);
        \fill[zero] (10,-1) rectangle (11,0);
        \fill[pole] (10, 0) rectangle (11,1);
        \draw[fill=pole, draw=none] (9,1) rectangle (10,2) 
        node[pos=.5]{$\times$};
        \draw[very thick] (0,0) rectangle (1,1);
        \draw[very thick] (1,0) rectangle (2,1);
        \draw[very thick] (2,0) rectangle (3,1);
        \draw[very thick] (3,0) rectangle (4,1);
        \draw[very thick] (4,0) rectangle (5,1);
        \draw[very thick] (5,0) rectangle (6,1);
        \draw[very thick, fill=pole] (6,0) rectangle (7,1) 
        node[pos=.5]{$\times$};
        \draw[very thick] (7,0) rectangle (8,1);
        \draw[very thick] (8,0) rectangle (9,1);
        \draw[very thick, fill=zero] (9,0) rectangle (10,1);
        \draw[very thick] (0,-1) rectangle (1,0);
        \draw[very thick] (1,-1) rectangle (2,0);
        \draw[very thick, fill=pole]  (2,-1) rectangle (3,0) 
        node[pos=.5]{$\times$};
        \draw[very thick] (3,-1) rectangle (4,0);
        \draw[very thick] (4,-1) rectangle (5,0);
        \draw[very thick] (5,-1) rectangle (6,0);
        \draw[very thick, fill=zero] (6,-1) rectangle (7,0);
        \draw[very thick] (0,-2) rectangle (1,-1);
        \draw[very thick] (1,-2) rectangle (2,-1);
        \draw[very thick] (2,-2) rectangle (3,-1);
        \draw[very thick, fill=pole] (0,-3) rectangle (1,-2) 
        node[pos=.5]{$\times$};
        \draw[very thick] (1,-3) rectangle (2,-2);
        \draw[very thick, fill=zero] (2,-3) rectangle (3,-2);
        \draw[very thick] (0,-4) rectangle (1,-3);
        \draw[very thick] (0,-5) rectangle (1,-4);
        \draw[very thick, fill=zero] (0,-6) rectangle (1,-5);
      \end{tikzpicture}
    \end{align}
    As we have seen before in claim \ref{claim:se}, evaluating a residue at the
    poles marked with a cross does not contribute to the integral. It is clear
    that placing a
    residue box as in \eqref{eq:residuebox}
    at one of the remaining poles in diagram \eqref{eq:diagramstep} yields again a
    Young diagram with the same structure of the poles and zeros as in diagram
    \eqref{eq:diagramstep}. 

    Hence we have established that all residues are simple residues located at
    \eqref{eq:unorderedresidues}. It is clear that different procedures to
    pick residues at $\{ z_s^{\alpha} : s \in Y_\alpha, \alpha=1,\dots,r\}$
    correspond to permutation of the variables $\hat{z}_j, j = 1 ,\dots, n$.
    Since the integral is symmetric in $z_1, \dots, z_n$, all possible
    permutations occur. This cancels the factor $n!$ in front of the integral.
    
    We have proved equation
    \eqref{eq:residuesum}. Next, we evaluate the
    right hand side of equation \eqref{eq:residueexpression}.
    The integrand (\ref{eq:zdef}) is invariant under swapping $q_1$ and
    $q_2$. The residues $\{z^\alpha_{s} : s \in Y_\alpha, \alpha = 1,\dots, r\}$
    defined in equation \eqref{eq:residuedef} remain invariant if we swap $q_1$
    and $q_2$ and transpose the diagrams $Y_1, \dots, Y_r$. Hence 
    \begin{align*}
      \sum_{|\vec{Y}|=n} \ZZ_{Y_1,\dots,Y_r}(\vec{u};\vec{p}) = 
      \sum_{|\vec{Y}|=n} \ZZ_{Y^T_1, \dots
        Y^T_r}(\vec{u};\vec{p})|_{q_1
        \leftrightarrow q_2}.
    \end{align*}
    We have $a_{Y}(x,y) = l_{Y^T}(y,x)$ and hence equation
    \eqref{eq:secondversion} follows from equation \eqref{eq:residueevaluation}.

    The proof of equation \eqref{eq:residueevaluation} is adapted from \cite{yanagida2010}, where it
    was performed for the special case $r=2$. It
    suffices to consider $\vec{p}=\es$ since we can cancel the factors 
    in the equation we want to
    prove. Let 
    \begin{align*}
      \RR_{\vec{Y}}(\vec{u};\es) = 
      \prod_{\alpha,\beta = 1}^{r} 
      \Bigg( &
      \prod_{s \in Y_\alpha}
      \frac{1}
      {1 - u_\alpha u_\beta^{-1} q_1^{l_{Y_\alpha}(s) + 1} 
        q_2^{-a_{Y_\beta}(s)} }
      \\ &\prod_{t \in Y_\beta} 
      \frac{1}
      {1 - u_\alpha u_\beta^{-1} q_1^{-l_{Y_\beta}(t)} q_2^{a_{Y_\alpha}(t)+1} }
      \Bigg)
    \end{align*}
    be the right hand side of the first equation in Theorem
    \ref{thm:residuecalculation}.

    For any $r$-tuple $\vec{Y}$ of partitions with $Y_r \neq \es$, define the
    $r$-tuple $\vec{Y}'$ of partitions by removing the last box from the last
    partition in $\vec{Y}$, i.e., we set $Y'_\alpha = Y_\alpha$ for $\alpha = 1,
    \dots, r-1$ and $Y'_r = (Y_1, \dots, Y_{l-1}, Y_l - 1)$, where $l$ is the
    length of $Y_r$. In terms of Young diagrams, we go from $\vec{Y}$ to
    $\vec{Y}'$ by removing the box \begin{align*}(l,w) := (l(Y_r),
      Y_r(l)) \end{align*}from
    the last Young diagram $Y_r$ in $\vec{Y}$. We will prove
    \begin{align}
      \frac{\ZZ_{\vec{Y}}(\vec{u};\es)}{\ZZ_{\vec{Y}'}(\vec{u};\es)}
      =\frac{\RR_{\vec{Y}}(\vec{u};\es)}{\RR_{\vec{Y}'}(\vec{u};\es)}
      \label{eq:step}.
    \end{align}
    This already suffices: Using equation (\ref{eq:step}), we can reduce the
    statement of the theorem to the case $\vec{Y} = (Y_1,\dots,Y_{r-1},\es)$. Both
    $\ZZ_{\vec{Y}}(\vec{u};\es)$ and $\RR_{\vec{Y}}(\vec{u};\es)$ are invariant under
    simultaneous permutation of the components of $\vec{Y}=(Y_1,\dots,Y_r)$ and
    $u = (u_1,\dots,u_r)$. This follows directly from the respective definitions.
    Hence, we can reduce the statement of the theorem to the case $\vec{Y} = (Y_1,
    \dots, Y_{r-2}, \es, Y_r)$ and, again using equation (\ref{eq:step}), to the
    case $\vec{Y} = (Y_1, \dots, Y_{r-2}, \es, \es)$. Continuing in this fashion,
    we can reduce the statement to the case $\vec{Y} = (\es, \dots, \es)$, in
    which it holds trivially.

    In the calculation of both sides of equation \eqref{eq:step} we have to
    evaluate telescopic products. In order to group the factors for such
    evaluations, we will have to keep track when $Y_\alpha(x)$ and $Y_\alpha^T(y)$
    remain constant as we vary the row and column indices. We write
    \begin{align*}
      Y_\alpha = (Y_\alpha(1),\dots,Y_\alpha(l(Y_\alpha))) 
      = ( \underbrace{F_\alpha(1), \dots, F_\alpha(1)}_{G_\alpha(1) \text{times}},
      \dots, 
      \underbrace{F_\alpha(m_\alpha), \dots,
        F_\alpha(m_\alpha)}_{G_\alpha(m_\alpha) \text{times}}) 
    \end{align*}
    We set $F_\alpha(m_\alpha+1) = 0$. Note that for any $x,y \in \N$,
    \begin{align*}
      Y_\alpha(y) & \in \{F_\alpha(j) : j = 1, \dots, m_\alpha + 1 \} \\
      Y_\alpha^T(x) & \in \{ G_\alpha(1) + \cdots + G_\alpha(j) : j = 0, \dots,
      m_\alpha \}.
    \end{align*}
    Define the index $j_\alpha$ by 
    \begin{align*}
      G_\alpha(1) + \cdots + G_\alpha(j_\alpha -1) < l \leq G_\alpha(1) + \cdots
      G_\alpha(j_\alpha)
    \end{align*}
    when this condition can be satisfied and $j_\alpha = m_\alpha+1$ otherwise. We
    also introduce the notation $H_\alpha(j) = G_\alpha(1) + \cdots + G_\alpha(j)$.
    We split products over rows of Young diagrams as
    \begin{align}
      \label{eq:splitx}
      \prod_{x=1}^l = \prod_{j=1}^{j_\alpha - 1}
      \prod_{x=H_\alpha(j-1) + 1}^{H_\alpha(j)}
      \times
      \prod_{x=H_\alpha(j_\alpha -1) + 1}^l
    \end{align}
    When $x$ comes from the product with index $j\in \{1, \dots, j_\alpha-1\}$, we
    have $Y_\alpha(y) = F_\alpha(j)$. In the remaining product, we have
    $Y_\alpha(y) = F_\alpha(j_\alpha)$. Products over columns of Young diagrams are
    grouped as
    \begin{align}
      \label{eq:splity}
      \prod_{y=1}^{Y_\alpha(l)} = \prod_{y=1}^{F_\alpha(j_\alpha)} =
      \prod_{j=j_\alpha}^{m_\alpha} \prod_{y=F_\alpha(j+1)+1}^{F_\alpha(j)}
    \end{align}
    where we have $Y^T_\alpha(y) = H_\alpha(j)$ if $y$ comes from
    the factor with value $j$. 

    The right hand side of equation (\ref{eq:step}) equals
    \begin{align*}
      \frac{\RR_{\vec{Y}}(\vec{u};\es)}{\RR_{\vec{Y}'}(\vec{u};\es)}
      &=
      \prod_{\alpha,\beta = 1}^{r} 
      \frac
      {
        \prod_{s \in Y'_\alpha}
        1 - \frac{u_\alpha}{ u_\beta} q_1^{l_{Y'_\alpha}(s) +
          1}q_2^{-a_{Y'_\beta}(s)} 
        \prod_{t \in Y'_\beta} 
        1 - \frac{u_\alpha}{ u_\beta} q_1^{-l_{Y'_\beta}(t)}
        q_2^{a_{Y'_\alpha}(t)+1} 
      }
      {
        \prod_{s \in Y_\alpha}
        1 - \frac{u_\alpha}{ u_\beta} q_1^{l_{Y_\alpha}(s)+ 1}
        q_2^{-a_{Y_\beta}(s)} 
        \prod_{t \in Y_\beta} 
        1 - \frac{u_\alpha}{ u_\beta} q_1^{-l_{Y_\beta}(t)} 
        q_2^{a_{Y_\alpha}(t)+1} 
      }.
    \end{align*}
    We introduce a variable $\xi$ to be able to ignore poles during the calculation.
    Regrouping we get
    \begin{align*}
      \frac{\RR_{\vec{Y}}(\vec{u};\es)}{\RR_{\vec{Y}'}(\vec{u};\es)}
      &=
      \lim_{\xi \to 1}
      S(\xi) 
      \prod_{\alpha=1}^{r-1} T_\alpha(\xi) U_\alpha(\xi)
    \end{align*}
    where
    \begin{align*}
      S(\xi) 
      =&
      \frac
      {
        1
      }
      {
        (\xi - q_1 )
        (\xi - q_2 )
      }
      \prod_{t \in Y'_r} 
      \frac
      {
        (\xi - q_1^{l_{Y'_r}(t) + 1}q_2^{-a_{Y'_r}(t)} )
        (\xi - q_1^{-l_{Y'_r}(t)} q_2^{a_{Y'_r}(t)+1} )
      }
      {
        (\xi - q_1^{l_{Y_r}(t) + 1}q_2^{-a_{Y_r}(t)} )
        (\xi - q_1^{-l_{Y_r}(t)} q_2^{a_{Y_r}(t)+1} )
      }
      \\
      T_\alpha(\xi)
      =&
      \frac
      {
        1
      }
      {
        \xi - u_\alpha u_r^{-1} q_1q_2^{a_{Y_\alpha}(l,w)+1} 
      }
      \prod_{t \in Y_r} 
      \frac
      {
        \xi - u_\alpha u_r^{-1} q_1^{-l_{Y'_r}(t)} q_2^{a_{Y_\alpha}(t)+1} 
      }
      {
        \xi - u_\alpha u_r^{-1} q_1^{-l_{Y_r}(t)} q_2^{a_{Y_\alpha}(t)+1} 
      }
      \\ &\times\prod_{s \in Y_\alpha}
      \frac
      {
        \xi - u_\alpha u_r^{-1} q_1^{l_{Y_\alpha}(s) + 1}q_2^{-a_{Y'_r}(s)} 
      }
      {
        \xi - u_\alpha u_r^{-1} q_1^{l_{Y_\alpha}(s) + 1}q_2^{-a_{Y_r}(s)} 
      }
      \\
      U_\alpha(\xi)
      =&
      \frac
      {
        1
      }
      {
        \xi - u_r u_\alpha^{-1} q_2^{-a_{Y_\alpha}(l,w)} 
      }
      \prod_{t \in Y_r}
      \frac
      {
        \xi - u_r u_\alpha^{-1} q_1^{l_{Y'_r}(t) + 1}q_2^{-a_{Y_\alpha}(t)} 
      }
      {
        \xi - u_r u_\alpha^{-1} q_1^{l_{Y_r}(t) + 1}q_2^{-a_{Y_\alpha}(t)} 
      }
      \\ &\times
      \prod_{s \in Y_\alpha} 
      \frac
      {
        \xi - u_r u_\alpha^{-1} q_1^{-l_{Y_\alpha}(s)} q_2^{a_{Y'_r}(s)+1} 
      }
      {
        \xi - u_r u_\alpha^{-1} q_1^{-l_{Y_\alpha}(s)} q_2^{a_{Y_r}(s)+1} 
      }
    \end{align*}
    Using
    \begin{align*}
      a_{Y'_\alpha}(x,y) & =
      \begin{cases}
        Y_\alpha(x)-y-1, & x = l, \alpha = r \\
        Y_\alpha(x)-y, & \text{otherwise} 
      \end{cases}, \\
      l_{Y'_\alpha}(x,y) & = 
      \begin{cases}
        Y_\alpha^T(y)-x-1, & y = w, \alpha = r \\
        Y_\alpha^T(y)-x, & \text{otherwise}
      \end{cases}
    \end{align*}
    and the splitting described in equations \eqref{eq:splitx} and \eqref{eq:splity}
    we get
    \begin{align*}
      S(\xi) 
      =&
      \frac
      {
        1
      }
      {
        (\xi - q_1q_2^{-w+1} )
        (\xi - q_2^{w} )
      }
      \frac
      {
        (\xi - q_1)
        (\xi - q_2)
      }
      {
        (\xi - 1)
        (\xi - q_1q_2 )
      }
      \\ &\prod_{j=1}^{m_r}
      \frac
      {
        (\xi - q_1^{l-H_r(j)}q_2^{-F_r(j)+w} )
        (\xi -  q_1^{-l+H_r(j)+1}q_2^{F_r(j)-w+1} )
      }
      {
        (\xi - q_1^{l-H_r(j-1)}q_2^{-F_r(j)+w} )
        (\xi - q_1^{-l+H_r(j-1)+1} q_2^{F_r(j)-w+1} )
      }, \\
      T_\alpha(\xi)
      =&
      \frac
      {
        1
      }
      {
        \xi - u_\alpha u_r^{-1} 
        q_1^{l(Y_\alpha)-l + 1}q_2^{-w+1}
      }
      \prod_{j=1}^{m_\alpha}
      \frac
      {
        \xi - u_\alpha u_r^{-1} q_1^{-l+H_\alpha(j)+1} q_2^{F_\alpha(j)-w+1} 
      }
      {
        \xi - u_\alpha u_r^{-1} q_1^{-l+H_\alpha(j-1)+1} q_2^{F_\alpha(j)-w+1} 
      }, \\
      U_\alpha(\xi)
      =&
      \frac
      {
        1
      }
      {
        \xi - u_r u_\alpha^{-1} 
        q_1^{-l(Y_\alpha)+l} q_2^{w}
      }
      \prod_{j=1}^{m_\alpha}
      \frac
      {
        \xi - u_r u_\alpha^{-1} q_1^{l-H_\alpha(j)}q_2^{-F_\alpha(j)+w} 
      }
      {
        \xi - u_r u_\alpha^{-1} q_1^{l-H_\alpha(j-1)}q_2^{-F_\alpha(j)+w} 
      }.
      \\
    \end{align*}
    Together
    \begin{align*}
      &\frac{\RR_{\vec{Y}}(\vec{u};\es)}{\RR_{\vec{Y}'}(\vec{u};\es)}
      =
      \lim_{\xi \to 1}
      \frac
      {
        (\xi - q_1)
        (\xi - q_2)
      }
      {
        (\xi - 1)
        (\xi - q_1q_2 )
      }\\
      &\prod_{\alpha=1}^{r} 
      \Bigg(
      \frac
      {
        1     
      }
      {
        (\xi - u_\alpha u_r^{-1}
        q_1^{l(Y_\alpha)-l + 1} q_2^{-w+1})
        (\xi - u_r u_\alpha^{-1}     q_1^{-l(Y_\alpha)+l} q_2^{w})
      }
      \\
      & \times
      \prod_{j=1}^{m_\alpha}
      \frac
      {
        (\xi - u_r u_\alpha^{-1} q_1^{l-H_\alpha(j)}q_2^{-F_\alpha(j)+w} )
        (\xi - u_\alpha u_r^{-1} q_1^{-l+H_\alpha(j)+1} q_2^{F_\alpha(j)-w+1} )
      }
      {
        (\xi - u_r u_\alpha^{-1} q_1^{l-H_\alpha(j-1)}q_2^{-F_\alpha(j)+w} )
        (\xi - u_\alpha u_r^{-1} q_1^{-l+H_\alpha(j-1)+1} q_2^{F_\alpha(j)-w+1} )
      }
      \Bigg)
    \end{align*}

    For the residue calculation, fix the order of the variables such that the
    integration over $z_n$ picks up the residue $z^r_{l,w}$ coming from the box
    $(l,w) \in Y_r$ we remove from the last partition in $\vec{Y}$ to get
    $\vec{Y}'$.  The left hand side of equation (\ref{eq:step}) equals
    \begin{align*}
      \frac{\ZZ_{\vec{Y}}(\vec{u};\es)}{\ZZ_{\vec{Y}'}(\vec{u};\es)}
      & = 
      \frac{1-q_1q_2}{(1-q_1)(1-q_2)}
      \;
      \lim_{z_j \to \hat{z}_j, j = 1, \dots, n}
      \;
      (1 -
      \frac
      {
        z^r_{l,w}
      }
      {
        z_n
      }
      )
      \frac
      {
        \mathcal{I}(z_1, \dots, z_n;\vec{u})
      }
      {
        \mathcal{I}(z_1, \dots, z_{n-1};\vec{u})
      }
      \\
    \end{align*}
    We take the first $(n-1)$ limits separately:
    The quotient  $\frac
    {
      \mathcal{I}(z_1, \dots, z_n;\vec{u})
    }
    {
      \mathcal{I}(z_1, \dots, z_{n-1};\vec{u})
    }$
    converges to
    \begin{align*}
      \prod_{\alpha=1}^r  
      \Bigg( &
      \frac{-u_\alpha  z_n }
      {
        (z_n - u_\alpha) 
        (q_1 q_2 z_n - u_\alpha)
      }\\ \times &
      \prod_{s \in Y'_\alpha}
      \frac
      {
        (z_n- z^\alpha_s)^2(z_n - q_1q_2z^\alpha_s)
        (z_n - q_1^{-1}q_2^{-1}z^\alpha_s)
      }
      {
        (z_n - q_1z^\alpha_s)(z_n - q_2z^\alpha_s)
        (z_n - q_1^{-1}z^\alpha_s)(z_n - q_2^{-1}z^\alpha_s)
      }
      \Bigg)
    \end{align*}
    for $z_j \to \hat{z}_j, j =1, \dots, n-1$. The factors with $\alpha \neq r$
    do not have poles for $z_n \to z_{l,w}^r$ since $u_\alpha/u_r \notin \{
    q_1^x q_2^y : x,y \in \Z \}$.  
    We define $\xi = \frac{z_n}{z^r_{l,w}}$ and set
    \begin{align*}
      A_\alpha(\xi)
      & :=
      \prod_{s \in Y_\alpha}
      \frac
      {
        (\xi - \frac{z_s^\alpha}{z^r_{l,w}})
        (\xi - q_1q_2 \frac{z_s^\alpha}{z^r_{l,w}})
        (\xi - \frac{z_s^\alpha}{z^r_{l,w}})
        (\xi - q_1^{-1} q_2^{-1}  \frac{z_s^\alpha}{z^r_{l,w}})
      }
      {
        (\xi - q_1 \frac{z_s^\alpha}{z^r_{l,w}})
        (\xi - q_2 \frac{z_s^\alpha}{z^r_{l,w}})
        (\xi - q_1^{-1} \frac{z_s^\alpha}{z^r_{l,w}})
        (\xi - q_2^{-1} \frac{z_s^\alpha}{z^r_{l,w}})
      } .
    \end{align*}
    Hence, the remaining limit $z_n \to z^r_{l,w}$ is given by
    \begin{align*}
      \frac{\ZZ_{\vec{Y}}(\vec{u};\es)}{\ZZ_{\vec{Y}'}(\vec{u};\es)}
      & = 
      \lim_{\xi \to 1}
      \frac
      {
        (\xi - q_1^{-1})
        (\xi - q_2^{-1})
      }
      {
        (\xi- 1)
        (\xi - q_1^{-1}q_2^{-1})
      }
      \prod_{\alpha=1}^r  
      \Bigg(
      \frac{ - \xi u_\alpha z^r_{l,w} }
      {
        (\xi z^r_{l,w} - u_\alpha) 
        (q_1 q_2 \xi z^r_{l,w} - u_\alpha)
      }
      A_\alpha(\xi)
      \Bigg),
    \end{align*}
    Using the splitting described in \eqref{eq:splitx} and \eqref{eq:splity} we get
    \begin{align*}
      A_\alpha(\xi) & = 
      \frac
      {
        (\xi - \frac{u_\alpha}{u_r} q_1^{1-l} q_2^{1-w} )
        (\xi - \frac{u_\alpha}{u_r} q_1^{-l} q_2^{-w} )
      }
      {
        (\xi - \frac{u_\alpha}{u_r} q_1^{l(Y_\alpha)-l+1} q_2^{1-w} )
        (\xi - \frac{u_\alpha}{u_r} q_1^{l(Y_\alpha)-l} q_2^{-w} )
      }
      \\
      & \times
      \prod_{j=1}^{m_\alpha}
      \frac
      {
        (\xi - \frac{u_\alpha}{u_r} q_1^{H_\alpha(j)-l+1} q_2^{F_\alpha(j)-w+1} )
        (\xi - \frac{u_\alpha}{u_r} q_1^{H_\alpha(j)-l} q_2^{F_\alpha(j)-w} )
      }
      {
        (\xi - \frac{u_\alpha}{u_r} q_1^{H_\alpha(j-1)-l+1} q_2^{F_\alpha(j)-w+1} )
        (\xi - \frac{u_\alpha}{u_r} q_1^{H_\alpha(j-1)-l} q_2^{F_\alpha(j)-w} )
      }
    \end{align*}
    Finally, we use $(\rho - \sigma) = -\rho \sigma (\rho^{-1} - \sigma^{-1})$
    repeatedly to conclude
    \begin{align*}
      \frac{\ZZ_{\vec{Y}}(\vec{u};\es)}{\ZZ_{\vec{Y}'}(\vec{u};\es)}
      & = 
      \lim_{\xi \to 1}
      \frac
      {
        (\xi - q_1^{-1})
        (\xi - q_2^{-1})
      }
      {
        (\xi- 1)
        (\xi - q_1^{-1}q_2^{-1})
      }
      \\
      & 
      \prod_{\alpha=1}^r  
      \Bigg(
      \frac{ - \xi u_\alpha u_r^{-1} q_1^{-l} q_2^{-w} }
      {
        (\xi - \frac{u_\alpha}{u_r} q_1^{l(Y_\alpha)-l+1} q_2^{1-w} )
        (\xi - \frac{u_\alpha}{u_r} q_1^{l(Y_\alpha)-l} q_2^{-w} )
      }
      \\
       \times &
      \prod_{j=1}^{m_\alpha}
      \frac
      {
        (\xi - \frac{u_\alpha}{u_r} q_1^{H_\alpha(j)-l+1} q_2^{F_\alpha(j)-w+1} )
        (\xi - \frac{u_\alpha}{u_r} q_1^{H_\alpha(j)-l} q_2^{F_\alpha(j)-w} )
      }
      {
        (\xi - \frac{u_\alpha}{u_r} q_1^{H_\alpha(j-1)-l+1} q_2^{F_\alpha(j)-w+1} )
        (\xi - \frac{u_\alpha}{u_r} q_1^{H_\alpha(j-1)-l} q_2^{F_\alpha(j)-w} )
      }
      \Bigg)
      \\
      & = 
      \lim_{\xi \to 1}
      \frac
      {
        (\xi^{-1} - q_1)
        (\xi^{-1} - q_2)
      }
      {
        (\xi^{-1}- 1)
        (\xi^{-1} - q_1q_2)
      }
      \\
      & 
      \prod_{\alpha=1}^r  
      \Bigg(
      \frac
      {
        1
      }
      {
        (\xi - \frac{u_\alpha}{u_r} q_1^{l(Y_\alpha)-l+1} q_2^{1-w} )
        (\xi^{-1} - \frac{u_r}{u_\alpha} q_1^{-l(Y_\alpha)+l} q_2^{w} )
      }
      \\
       \times &
      \prod_{j=1}^{m_\alpha}
      \frac
      {
        (\xi - \frac{u_\alpha}{u_r} q_1^{H_\alpha(j)-l+1} q_2^{F_\alpha(j)-w+1} )
        (\xi^{-1} - \frac{u_r}{u_\alpha} q_1^{-H_\alpha(j)+l} q_2^{-F_\alpha(j)+w} )
      }
      {
        (\xi - \frac{u_\alpha}{u_r} q_1^{H_\alpha(j-1)-l+1} q_2^{F_\alpha(j)-w+1} )
        (\xi^{-1} - \frac{u_r}{u_\alpha} q_1^{-H_\alpha(j-1)+l}
        q_2^{-F_\alpha(j)+w} )
      }
      \Bigg)
      \\ &=
      \frac{\RR_{\vec{Y}}(\vec{u};\es)}{\RR_{\vec{Y}'}(\vec{u};\es)}
    \end{align*}
    since no factor $(\xi^{+1} - \cdots)$ vanishes in the limit.

  \bibliography{biblio}{}
  \bibliographystyle{plain}
\end{document}